\documentclass[graybox,envcountsame]{svmult}

\usepackage{mathptmx}       % selects Times Roman as basic font
\usepackage{helvet}         % selects Helvetica as sans-serif font
\usepackage{courier}        % selects Courier as typewriter font
\usepackage{type1cm}        % activate if the above 3 fonts are
\usepackage{makeidx}         % allows index generation
\usepackage{graphicx}        % standard LaTeX graphics tool
\usepackage{multicol}        % used for the two-column index
\usepackage[bottom]{footmisc}% places footnotes at page bottom

\usepackage{enumerate}       % used for lists

\usepackage{amsmath}
\usepackage{amssymb}
\usepackage{bbm}
\usepackage{wasysym}
\usepackage{stmaryrd}
\usepackage{url}
\let\phi\varphi
\let\epsilon\varepsilon
\let\theta\vartheta
\let\mathbb\mathbbm

\newcommand{\bvp}[2]{\boxed{\begin{array}{l}#1\\#2\end{array}}}
\newcommand{\der}{\partial}
\newcommand{\cum}{{\textstyle \varint}}
\newcommand{\integral}{\int}
\newcommand{\intf}{{\textstyle \int_{\mathcal{F}}}}

\newcommand{\N}{\mathbb{N}}
\newcommand{\Z}{\mathbb{Z}}
\newcommand{\Q}{\mathbb{Q}}
\newcommand{\R}{\mathbb{R}}
\newcommand{\C}{\mathbb{C}}

\newcommand{\Vd}{\mathcal{F}^{*}}

\newcommand{\f}{f}        % forcing function

\newcommand{\kron}[1]{\hat{#1}}

\newcommand{\be}{\beta}

\let\oldsetminus\setminus
\renewcommand{\setminus}{\mathord{\oldsetminus}}

\newcommand{\freemon}{\langle X \rangle}
\newcommand{\freealg}{K \freemon}
\newcommand{\red}[3]{{}_{#1}\!#2_{#3}}
\newcommand{\irrmod}{\freealg_\downarrow}
\newcommand{\nf}[1]{\mathop{\downarrow} #1}
\newcommand{\steprel}{\mathop{\rightarrow}}
\newcommand{\sredrel}{\mathop{\smash{\overset{\raisebox{-0.35ex}{$\scriptstyle
          +$}}{\rightarrow}}}}
\newcommand{\redrel}{\mathop{\smash{\overset{\raisebox{-0.25ex}{$\scriptstyle
          *$}}{\rightarrow}}}}
\newcommand{\credrel}{\mathop{\smash{\overset{\raisebox{-0.25ex}{$\scriptstyle
          *$}}{\leftarrow}}}}

\newcommand{\galg}{\mathcal{F}}
\newcommand{\newgalg}{\smash{\hat{\mathcal{F}}}}

\newcommand{\evl}{\textbf{\scshape\texttt e}}
\newcommand{\ini}{\textbf{\scshape\texttt j}}
\newcommand{\const}{\mathcal{C}}
\newcommand{\init}{\mathcal{I}}
\newcommand{\bspc}{\mathcal{B}}
\newcommand{\fri}[1]{#1^\blacklozenge}

\newcommand{\Ker}[1]{\operatorname{Ker}(#1)}

\newcommand{\Img}[1]{\operatorname{Im}(#1)}

\newcommand{\dirs}{\dotplus}
\newcommand{\gdiffop}{\galg[\der]}

\newcommand{\gintop}{\galg[\cum]}

\newcommand{\bndop}{(\Phi)}
\newcommand{\Hur}{H(K)}
\newcommand{\multfunc}{\galg^{\bullet}}
\newcommand{\intdiffopchar}[1]{\galg_{#1}[\der,\cum]}

\newcommand{\intdiffop}{\galg[\der,\cum]}

\newcommand{\orth}[1]{#1^\perp}
\newcommand{\intdiffopbas}{\galg^\#[\der,\cum]}

\newcommand{\exppol}{K[x,e^{Kx}]}

\newcommand{\ocum}{{\setbox0=%
    \hbox{$\textstyle{\scriptstyle-}{\varint}$}%
    \textstyle{\vcenter{\hbox{$\scriptstyle-$}}\kern-.5\wd0}%
    \!\varint}}
\newcommand{\cuma}{\cum^*}

\newcommand{\scond}{\mathopen{|}\Phi\mathclose{)}}
\newcommand{\allscond}{\mathopen{|}\multfunc\mathclose{)}}

\newcommand{\trp}[1]{#1^{\top}}
\newcommand{\act}{\cdot}

\newcommand{\free}{\galg\langle \der, \cum \rangle}
\newcommand{\freepol}{\newgalg\langle \der, \cum \rangle}
\newcommand{\green}{\mathfrak{g}}
\newcommand{\basis}{\mathfrak{l}}

\newcommand{\B}{\mathcal{B}} % shortcut

\newcommand{\intdiffpol}[2][\galg]{#1\{#2\}}
\newcommand{\intdiffpoltwo}{\galg\{u,v\}}
\newcommand{\newintdiffop}{\newgalg[\der,\cum]}
\newcommand{\term}[3]{\mathcal{T}_{#1}(#2 \cup #3)}
\newcommand{\canf}{\mathcal{R}}
\newcommand{\fcanf}{\mathcal{R}_0}
\newcommand{\polalg}{A_{\mathcal{V}}[X]}
\newcommand{\identities}{\mathcal{E}}
\newcommand{\tma}{\textsf{\small TH$\exists$OREM$\forall$}}
\newcommand{\mma}{\textsf{\small Mathematica}}
\newcommand{\ncalg}{\textsf{\small NCAlgebra}}
\newcommand{\maple}{\textsf{\small Maple}}

\newcommand{\grb}{Gr\"obner}

\newenvironment{rrproof}{\begin{proof}}{\qed\end{proof}}
\newenvironment{mmacode}{\begin{center}}{\end{center}}

\makeatletter
\renewcommand{\subsection}{\@startsection%
  {section}%
  {10}%
  {0em}%
  {\baselineskip}%
  {-\fontdimen2\font plus -\fontdimen3\font minus -\fontdimen4\font}%
  {\normalfont\normalsize\bfseries\itshape}}
\makeatother

\newcommand{\factor}[1]{[ #1 ]}
\newcommand{\nest}{\sqcup}

\spnewtheorem{rrexampleraw}[theorem]{Example}{\bfseries}{\mdseries}
\newenvironment{rrexample}{\begin{rrexampleraw}}{\qed\end{rrexampleraw}}
\spnewtheorem{fact}[theorem]{Fact}{\bfseries}{\mdseries}

\begin{document}

\title*{Symbolic Analysis for Boundary Problems:\\From Rewriting to Parametrized \grb\ Bases}
\titlerunning{Symbolic Analysis for Boundary Problems}

\author{Markus Rosenkranz \and Georg Regensburger \and Loredana Tec  \and Bruno Buchberger}
\authorrunning{M.~Rosenkranz, G.~Regensburger, L.~Tec, and B.~Buchberger}

\institute{Markus Rosenkranz \at School of Mathematics, Statistics and Actuarial Science (SMSAS), \\
University of Kent, Canterbury CT2 7NF, United Kingdom, \\
\email{M.Rosenkranz@kent.ac.uk}
\and Georg Regensburger \at Johann Radon Institute for Computational and Applied Mathematics (RICAM), \\
Austrian Academy of Sciences, 4040 Linz, Austria, \\
INRIA Saclay -- \^{I}le de France, Project DISCO, L2S, Sup\'{e}lec, 91192 Gif-sur-Yvette Cedex, France, \\
\email{Georg.Regensburger@ricam.oeaw.ac.at}\\
\and Loredana Tec \at Research Institute for Symbolic Computation (RISC),\\
Johannes Kepler University, 4032 Castle of Hagenberg, Austria, \\
\email{ltec@risc.uni-linz.ac.at}
\and Bruno Buchberger \at Research Institute for Symbolic Computation (RISC),\\
Johannes Kepler University, 4032 Castle of Hagenberg, Austria, \\
\email{Bruno.Buchberger@risc.uni-linz.ac.at}
}

\maketitle

\abstract{We review our algebraic framework for linear boundary problems (concentrating on ordinary differential equations). Its starting point is an appropriate algebraization of the domain of functions, which we have named integro-differential algebras. The algebraic treatment of boundary problems brings up two new algebraic structures whose symbolic representation and computational realization is based on canonical forms in certain commutative and noncommutative polynomial domains. The first of these, the ring of integro-differential operators, is used for both stating and solving linear boundary problems. The other structure, called integro-differential polynomials, is the key tool for describing extensions of integro-differential algebras. We use the canonical simplifier for integro-differential polynomials for generating an automated proof establishing a canonical simplifier for integro-differential operators. Our approach is fully implemented in the \tma\ system;  some code fragments and sample computations are included.}

\section{Introduction}
\label{sec:intro}

\subsection{Overall View.}

When problems from Analysis---notably differential equations---are treated by methods from Symbolic Computation, one speaks of \mbox{\emph{Symbolic Analysis}}, as in the eponymous workshops of the \textsf{FoCM} conference series~\cite{CuckerShub1997}. Symbolic Analysis is based on algebraic structures, as all other symbolic branches, but its special flavor comes from its connection with analytic and numeric techniques. As most differential equations arising in the applications can only be solved numerically, this connection is absolutely vital.

If symbolic techniques cannot solve ``most'' differential equations, \emph{what else} can they do? The answers are very diverse (reductions, normal forms, symmetry groups, singularity analysis, triangularization etc), and in the frame of this paper we can only point to surveys like~\cite{Seiler1997} and~\cite[\S2.11]{GrabmeierKaltofenWeispfenning2003}. In fact, even the notion of ``solving'' is quite subtle and can be made precise in various ways. Often a symbolic method will not provide the ``solution'' in itself but valuable information about it to be exploited for subsequent numerical simulation.

Our own approach takes a somewhat intermediate position while diverging radically in another respect: Unlike most other symbolic methods known to us, we consider differential equations along with their \emph{boundary conditions}. This is not only crucial for many applications, it is also advantageous from an algebraic point of view: It allows to define a linear operator, called the Green's operator, that maps the so-called forcing function on the right-hand side of an equation to the unique solution determined by the boundary conditions. This gives rise to an interesting structure on Green's operators and on boundary problems (Sect.~\ref{sec:bp}). Algebraically, the consequence is that we have to generalize the common structure of differential algebras to what we have called integro-differential algebras (Sect.~\ref{sec:intdiffalg}).

Regarding the \emph{solvability issues}, the advantage of this approach is that it uncouples the task of finding an algebraic representation of the Green's operator from that of carrying out the quadratures involved in applying the Green's operator to a forcing function. While the latter may be infeasible in a symbolic manner, the former can be done by our approach (with numerical quadratures for integrating forcing functions).

The \emph{research program} just outlined has been pursued in the course of the SFB project F013 (see below for a brief chronology), and the results have been reported elsewhere~\cite{Rosenkranz2005,RosenkranzRegensburger2008a,RegensburgerRosenkranz2009}. For the time being, we have restricted ourselves to linear boundary problems, but the structure of integro-differential polynomials~\cite{RosenkranzRegensburger2008}  may be a first stepping stone towards nonlinear Green's operators. Since the algebraic machinery for Green's operators is very young, our strategy was to concentrate first on boundary problems for ordinary differential equations (ODEs), with some first steps towards partial differential equations (PDEs) undertaken more recently~\cite{RosenkranzRegensburgerTecBuchberger2009}. For an application of our methods in the context of actuarial mathematics, we refer to~\cite{AlbrecherConstantinescuPirsicRegensburgerRosenkranz2010}, for a more algebraic picture from the skew-polynomial perspective see~\cite{RegensburgerRosenkranzMiddeke2009}.

\subsection{New Results.}

In the present paper, we will present a \emph{new confluence proof} for the central data structure used in our approach: As the algebraic language for Green's operators, the integro-differential operators (Sect.~\ref{sec:intdiffop}) are defined as a ring of noncommutative polynomials in infinitely many variables, modulo an infinitely generated ideal. While the indeterminates represent the basic operations of analysis (differentiation, integration, extraction of boundary values and multiplication by one of infinitely many coefficient functions), this ideal specifies their interaction (e.g. the fundamental theorem of calculus describing how differentiation and integration interact). Our new proof is fully automated within the \tma\ system (Sect.~\ref{sec:theorema}), using a generic noncommutative polynomial reduction based on a noncommutative adaption of reduction rings~\cite{Buchberger2001}; see also~\cite{TecRegensburgerRosenkranzBuchberger2010} for a short outline of the proof.

In a way, the new proof completes the circle started with the ad-hoc confluence proof in~\cite{Rosenkranz2003}. For the latter, no algebraic structure was available for coping with certain expressions that arise in the proof because they involved generic coefficient functions along with their integrals and derivatives (rather than the operator indeterminates modeling integration and differentiation!), while this structure is now provided by the afore-mentioned \emph{integro-differential polynomials} (Sect.~\ref{sec:intdiffpol}). Roughly speaking, this means within the spectrum between rewrite systems (completion by the Knuth-Bendix procedure) and \grb\ bases (completion by Buchberger's algorithm), we have moved away from the former towards the latter~\cite{Buchberger1987}. We will come back to this point later (Sect.~\ref{sec:rwtogb}).

Moreover, the paper includes the following \emph{improvements and innovations}: The setting for \grb\ bases and the Buchberger algorithm are introduced generically for commutative and noncommutative rings (allowing infinitely many variables and generators), based on reduction rings and implemented in the \tma\ system (Sect.~\ref{sec:theorema}). The presentation of integro-differential algebras is streamlined and generalized (Sect.~\ref{sec:intdiffalg}). For both of the main computational domains---integro-differential operators and integro-differential polynomials---we have a basis free description while a choice of basis is only need for deciding equality (Sects.~\ref{sec:intdiffop},~\ref{sec:intdiffpol}). The construction of integro-differential polynomials, which was sketched in~\cite{RosenkranzRegensburger2008}, is carried out in detail (Sect.~\ref{sec:intdiffpol}). In particular, a complete proof of the crucial result on canonical forms (Thm.~\ref{thm:baxter}) is now given.

\subsection{Chronological Outline.}

As indicated above, this paper may be seen as a kind of target line for the research that we have carried out within Project F1322 of the SFB F013 supported by the Austrian Science Fund (FWF). We have already pointed out the crucial role of analysis/numerics in providing the right inspirations for the workings of Symbolic Analysis. The development of this project is an illuminating and pleasant case in point. It was initiated by the stimulating series of \emph{Hilbert Seminars} conducted jointly by Bruno Buchberger and Heinz W. Engl from October 2001 to July 2002, leading to the genesis of Project F1322 as a spin-off from Projects F1302 (Buchberger) and F1308 (Engl). Triggered by the paper~\cite{HeltonStankus2008}, the idea of symbolic operator algebras emerged as a common leading theme. It engendered a vision of transplanting certain ideas like the Moore-Penrose inverse on Hilbert spaces from their homeground in functional analysis into a new domain within Symbolic Analysis, where powerful algebraic tools like \grb\ bases are available~\cite{Buchberger1965,Buchberger1970,Bergman1978,Buchberger1998}. This vision eventually crystallized in the algebraic machinery for computing Green's operators as described before.

In the early stage of the project, those two main tools from analysis (Moore-Penrose inverse) and algebra (\grb\ bases) were welded together in a rather ad-hoc manner, but it did provide a new tool for solving boundary problems~\cite{RosenkranzBuchbergerEngl2003}. In the course of the \emph{dissertation}~\cite{Rosenkranz2003}, a finer analysis led to a substantial simplification where the Moore-Penrose inverse was superseded by a purely algebraic formulation in terms of one-sided inverses and the expensive computation of a new noncommutative \grb\ basis for each boundary problem was replaced by plain reduction modulo a fixed \grb\ basis for modeling the essential operator relations. The resulting quotient algebra (called ``Green's polynomials'' at that time) is the precursor of the integro-differential operators described below (Sect.~\ref{sec:intdiffop}). The final step towards the current setup was the reformulation and generalization in a differential algebra setting~\cite{RosenkranzRegensburger2008a} and in an abstract linear algebra setting~\cite{RegensburgerRosenkranz2009}.

The advances on the theoretical side were paralleled by an early \emph{implementation} of the algorithm for computing Green's operators. While the ad-hoc approach with computing \grb\ bases per-problem was carried out by the help of \ncalg, a dedicated \mma~package for noncommutative algebra~\cite{HeltonStankus2008}, the fixed \grb\ basis for simplifying Green's operator was implemented in the \tma\ system~\cite{BuchbergerCraciunJebeleanEtAl2006}; see Sect.~\ref{sec:theorema} for a general outline of this system. As the new differential algebra setting emerged, however, it became necessary to supplant this implementation by a new one. It was again integrated in the \tma\ system, but now in a much more intimate sense: Instead of using a custom-tailored interface as in~\cite{Rosenkranz2003}, the new package was coded directly in the \tma\ language using the elegant structuring constructs of functors~\cite{Buchberger2008}. Since this language is also the object language of the provers, this accomplishes the old ideal of integrating computation and deduction.

The presentation of several parts of this paper---notably Sects.~\ref{sec:intdiffalg}, \ref{sec:intdiffop}, \ref{sec:bp}---benefited greatly from a \emph{lecture} given in the academic year 2009/10 on Symbolic Integral Operators and Boundary Problems by the first two authors. The lecture was associated with the Doctoral Program ``Computational Mathematics: Numerical Analysis and Symbolic Computation'' (W1214), which is a follow-up program to the SFB F013. We would like to thank our students for the lively discussions and valuable comments.

\subsection{Overview of the Paper.}

We commence by having a closer look at the \tma\ system (Sect.~\ref{sec:theorema}), which will also be used in all sample computations presented in subsequent sections; both the sample computations and the \tma\ program code is available in an executable \mma\ notebook from \url{www.theorema.org}. We discuss canonical simplifiers for quotient structures and \grb\ bases in reduction rings, and we give a short overview of the functors used in building up the hierarchy of the algebraic structures used in the computations. The main structure among these is that of an integro-differential algebra (Sect.~\ref{sec:intdiffalg}), which is the starting point for the integro-differential operators as well as the integro-differential polynomials. Since the former are, in turn, the foundation for computing Green's operators for boundary problems, we will next summarize the construction of integro-differential operators and their basic properties (Sect.~\ref{sec:intdiffop}), while the algorithms for solving and factoring boundary problems are explained and exemplified thereafter (Sect.~\ref{sec:bp}). Driving towards the focus point of this paper, we describe then the algebra of integro-differential polynomials (Sect.~\ref{sec:intdiffpol}), which will be the key tool to be employed for the confluence proof. Since this proof is reduced to a computation in \tma, we will only explain the main philosophy and show some representative fragments (Sect.~\ref{sec:rwtogb}). We wind up with some thoughts about open problems and future work (Sect.~\ref{sec:conclusion}).

\section{Data Structures for Polynomials in Theorema}
\label{sec:theorema}

\subsection{The Theorema Functor Language.}
The \tma\ system~\cite{BuchbergerCraciunJebeleanEtAl2006} was designed by B.~Buchberger as an integrated environment for proving, solving and computing in various domains of mathematics. Implemented on top of \mma, its core language is a version of higher-order predicate logic that contains a natural programming language such that algorithms can be coded and verified in a unified formal frame. In this logic-internal programming language, functors are a powerful tool for building up \emph{hierarchical domains} in a modular and generic way. They were introduced and first implemented in \tma\ by B.~Buchberger. The general idea---and its use for structuring those domains in which \grb\ bases can be computed---is described in~\cite{Buchberger2001,Buchberger2008}, where one can also find references to pertinent early papers by B.~Buchberger. See also~\cite{Windsteiger1999} for some implementation aspects of functor programming.

The notion of functor in \tma\ is akin to functors in ML, not to be confused with the functors of category theory. From a computational point of view, a \tma\ functor is a higher-order function that produces a \emph{new domain from given domains}, where each domain is considered as a bundle of operations (including relations qua boolean-valued operations---in particular also carrier predicates). Operations in the new domain are defined in terms of operations in the underlying domains.

Apart from this computational aspect, functors also have an important reasoning aspect---a functor transports properties of the input domains to properties of the output domain, typical examples being the various ``preservation theorems'' in mathematics: ``If $R$ is a ring, then $R[x]$ is also a ring''. This means the functor $R \mapsto R[x]$ preserves the property of being a ring, in other words: it goes from the ``category of rings'' to itself. In this context, a \emph{category} is simply a collection of domains characterized by a common property (a higher-order predicate on domains).

See below for an example of a functor named $\mathtt{LexWords}$. It takes a linearly ordered alphabet $\mathtt{L}$ as input domain and builds the \emph{word monoid} over this alphabet:
\begin{mmacode}
\includegraphics[scale=0.75]{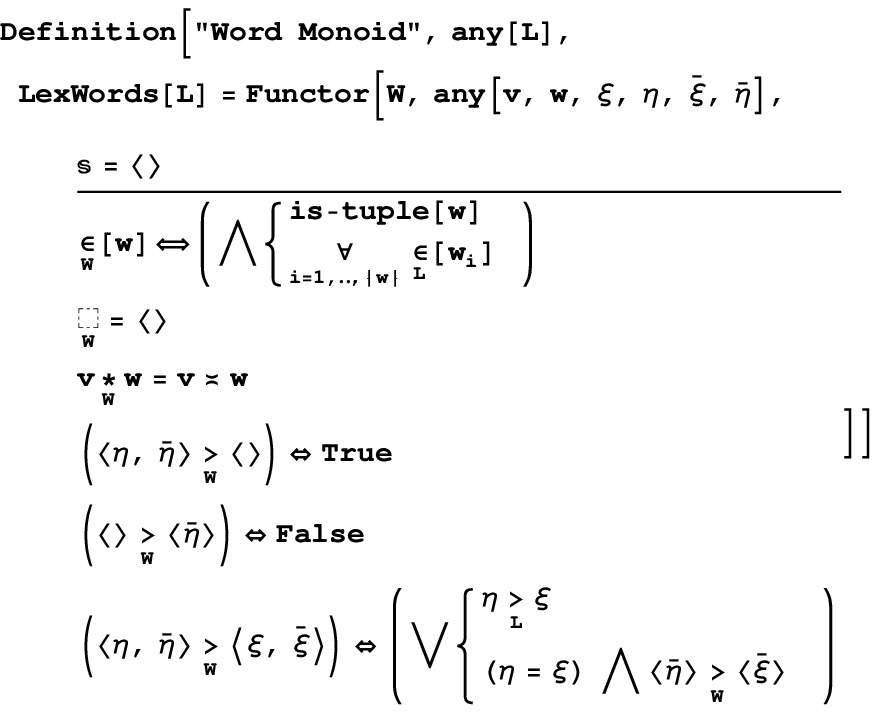}
\end{mmacode}
\noindent Here $\mathtt{\bar{\xi}}$, $\mathtt{\bar{\eta}}$ are sequence variables, i.e. they can be instantiated with finite sequences of terms. The new domain $\mathtt{W}$ has the following operations: $\mathtt{W[\in]}$ denotes the carrier predicate, the neutral element is given by $\mathtt{W[\Box]}$, the multiplication $\mathtt{W[\ast]}$ is defined as concatenation, and $\mathtt{W[>]}$ defines the lexicographic ordering on $\mathtt{W}$.

In the following code fragments, we illustrate one way of \emph{building up polynomials} in \tma\, starting from the base categories of fields with ordering and ordered monoids. Via the functor $\mathtt{FreeModule}$, we construct first the free vector space $\mathtt{V}$ over a field $\mathtt{K}$ generated by the set of words in an ordered monoid $\mathtt{W}$. The elements of $\mathtt{V}$ are described by $\mathtt{V}[\in]$ as lists of pairs, each pair containing one (non-zero) coefficient from $\mathtt{K}$ and one basis vector from $\mathtt{W}$, where the basis vectors are ordered according to the ordering on $\mathtt{W}$. The operations of addition, subtraction and scalar multiplication are defined recursively, using the operations on $\mathtt{K}$ and $\mathtt{W}$:
\begin{center}
\includegraphics[scale=0.75]{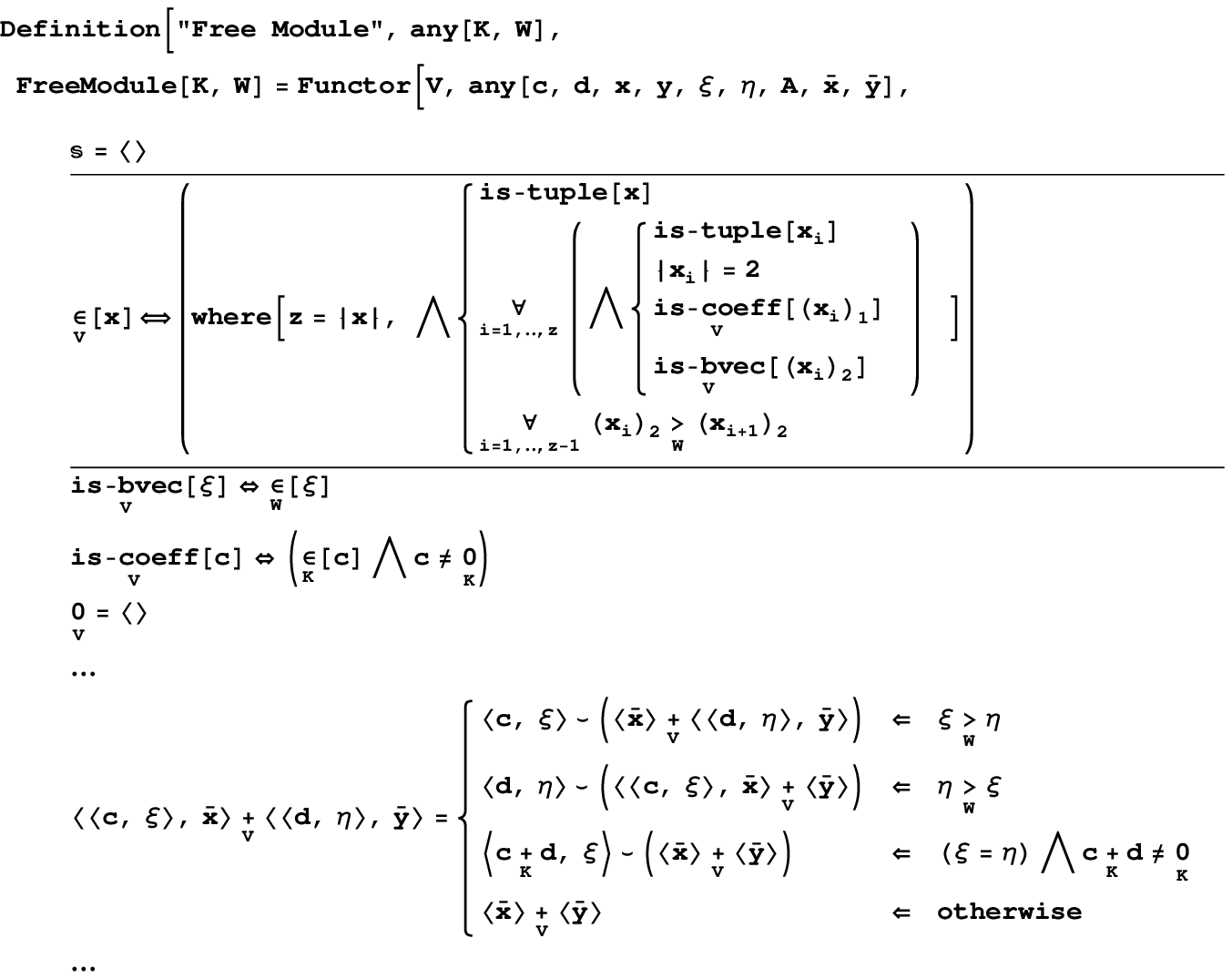}
\end{center}
\begin{center}
\includegraphics[scale=0.75]{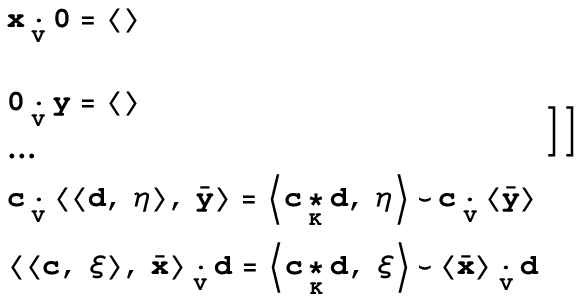}
\end{center}

By the $\mathtt{MonoidAlgebra}$ functor we extend this domain, introducing a multiplication using the corresponding operations in $\mathtt{K}$ and $\mathtt{W}$:
\begin{center}
\includegraphics[scale=0.75]{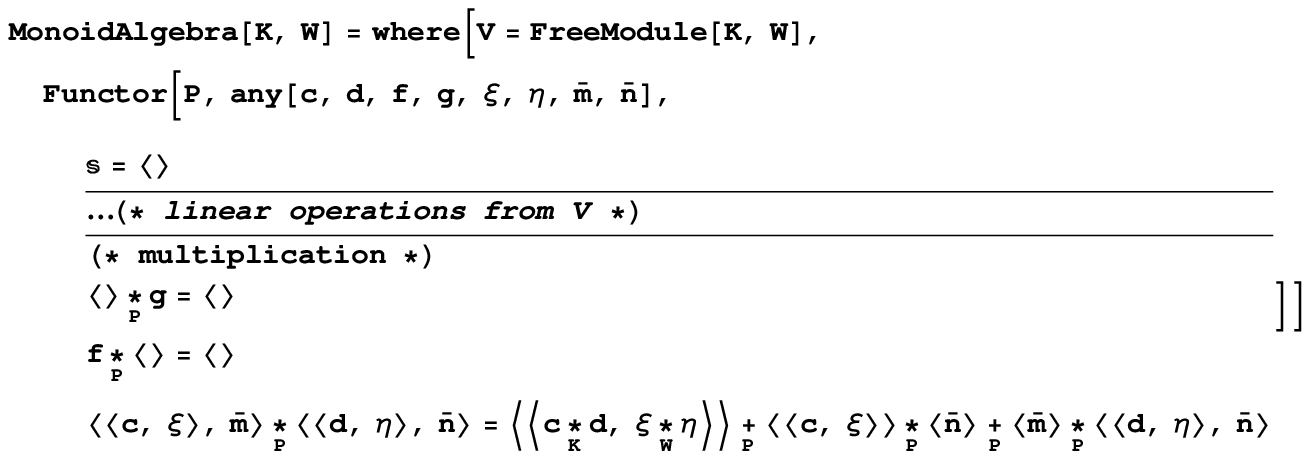}
\end{center}

\noindent The new domain inherits the structure on the elements of $\mathtt{V}$.

The main advantage of the above construction is that it is fully \emph{generic}: Not only can it be instantiated for different coefficient rings (or fields) and different sets of indeterminates, it comprises also the commutative and noncommutative case (where $\mathtt{W}$ is instantiated respectively by a commutative and noncommutative monoid).

\subsection{Quotient Structures and Canonical Simplifiers.}
In algebra (and also in the rest of mathematics), one encounters quotient structures on many occasions. The general setting is a set $A$ with various operations (an algebra in the general sense used in Sect.~\ref{sec:intdiffpol}) and a congruence relation $\equiv$ on $A$, meaning an equivalence relation that is compatible with all the operations on $A$. Then one may form the quotient $A/\mathord\equiv$, which will typically inherit some properties of $A$. For example, $A/\mathord\equiv$ belongs to the category of rings if $A$ does, so we can view the \emph{quotient construction} $A \mapsto A/\mathord\equiv$ as a functor on the category of rings.

But for computational purposes, the usual set-theoretic description of $A/\mathord\equiv$ as a set of equivalence classes is not suitable (since each such class is typically uncountably infinite). We will therefore use an alternative approach that was introduced in~\cite{BuchbergerLoos1983} as a general framework for symbolic representations. The starting point is a \emph{canonical simplifier} for $A/\mathord\equiv$, meaning a map $\sigma\colon A \rightarrow A$ such that
\begin{equation}
\label{eq:can-simpl}
\sigma(a) \equiv a \quad\text{and}\quad \sigma(a) = \sigma(a') \text{ whenever } a \equiv a'.
\end{equation}
The set $\tilde{A} = \sigma(A)$ is called the associated \emph{system of canonical forms} for $A/\mathord\equiv$.

Clearly canonical simplifiers exist for every quotient $A/\mathord\equiv$, but for computational purposes the crucial question is whether $\sigma$ is \emph{algorithmic}. Depending on $A/\mathord\equiv$, it may be easy or difficult or even impossible to construct a computable $\sigma\colon A \rightarrow A$. In the examples that we will treat, canonical simplifiers are indeed available.

Canonical simplifiers are also important because they allow us to \emph{compute in the quotient structure}. More precisely, one can transplant the operations on $A$ to $\tilde{A}$ by defining $\omega(a_1, \ldots, a_n) = \sigma(\omega(a_1, \ldots, a_n))$ for every operation $\omega$ on $A$. With these new operations, one may easily see that $\tilde{A}$ is isomorphic to the quotient $A/\mathord\equiv$; see the Theorem ``Canonical simplification and computation'' in~\cite[p.~13]{BuchbergerLoos1983}.

There is an intimate relation between canonical forms and \emph{normal forms} for rewrite systems (Sect.~\ref{sec:intdiffop} contains some basic terminology and references). In fact, every rewrite system $\mathord{\steprel}$ on an algebraic structure $A$ creates an equivalence relation $\equiv$, the symmetric closure of $\mathord{\redrel}$. Thus $a \equiv a'$ if and only if $a$ and $a'$ can be connected by an equational chain (using the rewrite rules in either direction). Typically, the relation $\equiv$ will actually be a congruence on $A$, so that the quotient $A/\mathord\equiv$ has a well-defined algebraic structure. Provided the rewrite system is noetherian, the normal forms of $\mathord{\steprel}$ are then also canonical forms for $A/\mathord\equiv$. Hence we will often identify these terms in a rewriting context.

For our \emph{implementation}, we use canonical simplifiers extensively. In fact, the observation made above about computing in the quotient structure is realized by a \tma\ functor, which is applied at various different places. Here $A$ is typically a $K$-algebra, with the ground field $K$ being $\Q$ or computable subfields of $\R$ and $\C$.

\subsection{Reduction Rings and \grb\ Bases.}
For defining reduction on polynomials, we use the \emph{reduction ring} approach in the sense of~\cite{Buc84,Buchberger2001}. For commutative reduction rings, see also~\cite{Stifter1987,Stifter1993}; for another noncommutative approach we refer to~\cite{MadlenerReinert1998b,MadlenerReinert1998a,MadlenerReinert1999}.

To put it simply, a reduction ring is a ring in which \grb\ bases can be done. A full \emph{axiomatization} for the commutative case is given in~\cite{Buc84}. If such rings satisfy certain additional axioms (defining the category of so-called ``\grb\ rings''), then \grb\ bases can be computed by iterated S-polynomial reduction in the given ring---this is the \grb\ Ring Extension Theorem, stated and proved in~\cite{Buc84}.

A detailed presentation of their \emph{construction} in the \tma\ setting was given in~\cite{Buchberger2001a,Buchberger2003}; it is the starting point for our current work. At this point we do not give an axiomatic characterization for noncommutative reduction rings, but we do use a construction that is similar to the commutative setting. Thus we endow a polynomial domain $\mathtt{P}$, built via the $\mathtt{MonoidAlgebra}$ functor with word monoid $\mathtt{W}$ and field $\mathtt{K}$, with the following three operations: a noetherian (partial) ordering, a binary operation \emph{least common reducible}, and a binary operation \emph{reduction multiplier}. The noetherian ordering is defined in the usual way in terms of the given orderings on $\mathtt{K}$ and $\mathtt{W}$.

The basic idea of \emph{reduction multipliers} is to answer the question: ``With which monomial do I have to multiply a given polynomial so that it cancels the leading term of another given polynomial?'' In the noncommutative case, the corresponding operation $\mathtt{rdm}$ splits into left reduction multiplier $\mathtt{lrdm}$ and its right counterpart $\mathtt{rrdm}$ defined as follows:
\begin{mmacode}
\includegraphics[scale=0.75]{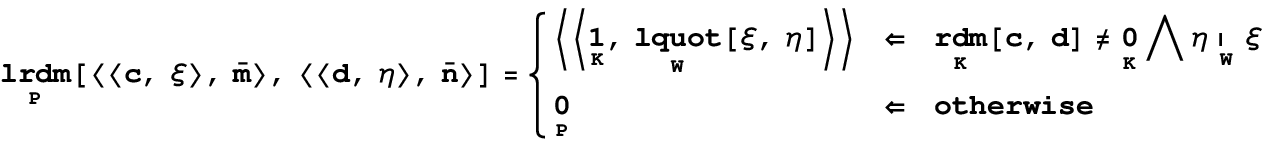}
\end{mmacode}
\begin{mmacode}
\includegraphics[scale=0.75]{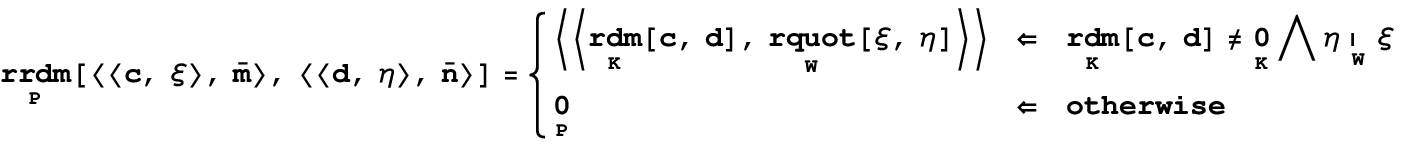}
\end{mmacode}

Here the divisibility relation $|$ on $\mathtt{W}$ checks whether a given word occurs within another word, and the corresponding quotients $\mathtt{lquot}$ and $\mathtt{rquot}$ yield the word segments respectively to the left and to the right of this occurrence. Since the scalars from $\mathtt{K}$ commute with the words, it is an arbitrary decision whether one includes it in the right (as here) or left reduction multiplier. In typical cases, this scalar factor is just $\mathtt{rdm[c,d] = c/d}$.

The operations relating \grb\ bases are introduced via a functor which is called
$\mathtt{GroebnerExtension}$. It defines \emph{polynomial reduction} using reduction multipliers (note that this includes also the commutative case, where one actually needs only one reduction multiplier, the other one being unity):

\begin{mmacode}
\includegraphics[scale=0.75]{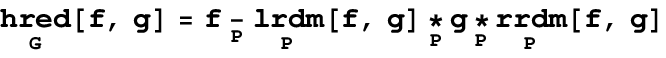}
\end{mmacode}
The next step is to introduce reduction modulo a system of
polynomials. For some applications (like the integro-differential
operators described in Sect.~\ref{sec:intdiffop}), it is necessary to
deal with \emph{infinite reduction systems}: the polynomial ring
contains infinitely many indeterminates, and reduction is applied
modulo an infinite set of polynomials. In other words, we want to deal
with an infinitely generated ideal in an infinitely generated algebra.

This is a broad topic, and we cannot hope to cover it in the present
scope. In general one must distinguish situations where both the
generators of the ideal and the algebra are parametrized by finitely
many families involving finitely many parameters and more general
algebras/ideals where this is not so. In the latter case, one works
with finite subsets, and all computations are approximate: one never
catches the whole algebraic picture. Fortunately, the applications we
have in mind---in particular the integro-differential operators---are
of the first type where \emph{full algorithmic control} can be
achieved. However, most of the common packages implementing
noncommutative \grb\ bases do not support such
cases~\cite{Levandovskyy2006,Levandovskyy2008}. For some recent
advances, we refer the reader
to~\cite{AschenbrennerHillar2008,BrouwerDraisma2009,
  HillarSullivant2009,LaScalaLevandovskyy2009} as well as Ufnarovski's
extensive survey chapter~\cite{Ufnarovskij1995}.

Let us point out just one important class of decidable reductions
in infinitely generated algebras---if an infinite set of (positively
weighted) \emph{homogeneous polynomials} is given, which is known to
be complete for each given degree (see~\cite{LaScalaLevandovskyy2009}
for the proof) since one can compute a truncated Gr\"obner basis of
such a graded ideal, which is finite up to a given degree. But if the
given set is not homogeneous or cannot be clearly presented degree by
degree, basically nothing can be claimed in general. Unfortunately,
the applications we have in mind seem to be of this type.

In our setting, infinitely generated ideals are handled by an
algorithmic operation for instantiating reduction rules. The
\emph{reduction} of polynomial $\mathtt{f}$ modulo a system
$\mathtt{S}$ is realized thus:

\begin{mmacode}
\includegraphics[scale=0.75]{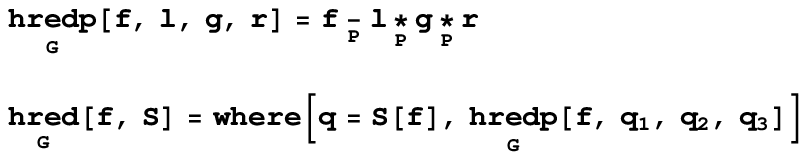}
\end{mmacode}

\noindent where $\mathtt{S[f]}$ is the operation that decides if there exists $\mathtt{g}$ modulo which $\mathtt{f}$ can be reduced, and it returns a triple containing the $\mathtt{g}$ and the left/right reduction multipliers needed for performing the reduction.

The main tool for the \grb\ bases construction, namely the notion of S-polynomial, can now be defined in terms of the \emph{least common reducible}:
\begin{mmacode}
\includegraphics[scale=0.75]{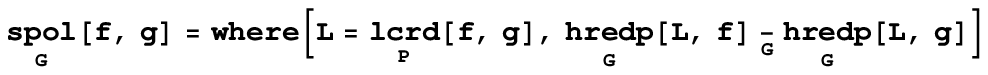}
\end{mmacode}
Here $\mathtt{lcrd[f,g]}$ represents the smallest monomial that can be reduced both modulo $\mathtt{f}$ and modulo $\mathtt{g}$, built from the least common reducible of the corresponding coefficients in $\mathtt{K}$ and the least common multiple of the words in $\mathtt{W}$:
\begin{mmacode}
\includegraphics[scale=0.75]{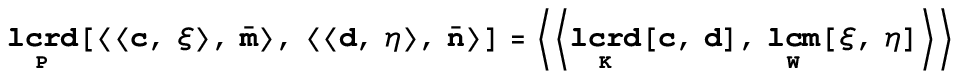}
\end{mmacode}
In our setting, the $\mathtt{lcrd[c,d]}$ can of course be chosen as unity since we work over a field $\mathtt{K}$, but in rings like $\Z$ one would have to use the least common multiple.

Finally, \emph{\grb\ bases} are computed by the usual accumulation of S-polynomials reduction, via the following version of Buchberger algorithm~\cite{Buchberger1965}:
\begin{mmacode}
\includegraphics[scale=0.75]{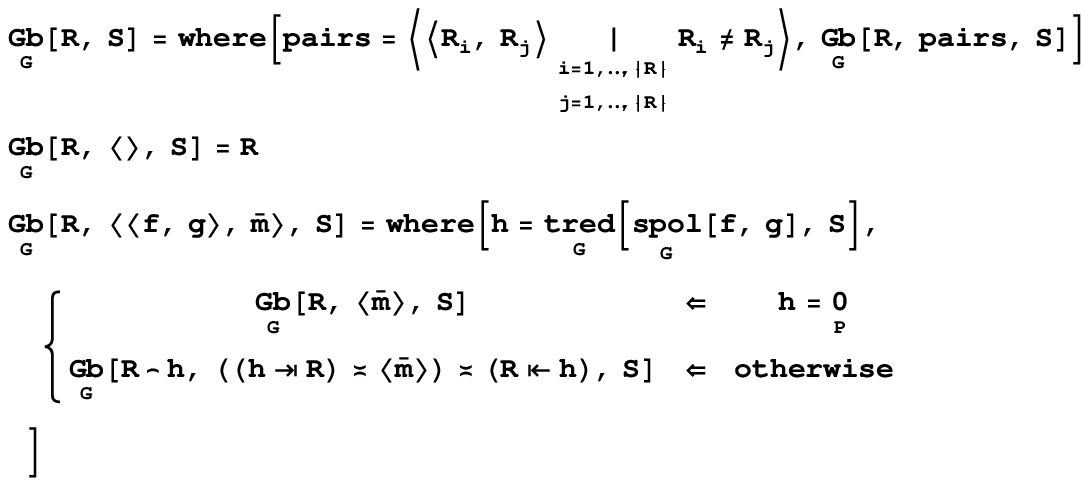}
\end{mmacode}
\noindent  Total reduction modulo a system, denoted here by $\mathtt{tred}$, is computed by iteratively performing reductions, until no more reduction is possible. The above implementation of Buchberger's algorithm is again \emph{generic} since it can be used in both commutative and noncommutative settings.  For finitely many indeterminates, the algorithm always terminates in the commutative case (by Dickson's Lemma); in the noncommutative setting, this cannot be guaranteed in general. For our applications we also have to be careful to ensure that the reduction systems we use are indeed noetherian (Sect.~\ref{sec:intdiffop}).

\section{Integro-Differential Algebras}
\label{sec:intdiffalg}

For working with boundary problems in a symbolic way, we first need an algebraic structure having \emph{differentiation along
with integration}. In the following definitions, one may think of our standard example $\galg = C^\infty (\mathbb{R})$,
where $\der=\,'$ is the usual derivation and $\cum$ the integral operator
\begin{equation*}
f \mapsto \int_a^x f(\xi) \, d\xi
\end{equation*}
for a fixed $a \in \mathbb{R}$.

\subsection{Axioms and Basic Properties.}

Let $K$ be a commutative ring. We first recall that $(\galg, \der)$ is a \emph{differential $K$-algebra}
if $\der\colon \galg \rightarrow \galg$ is a $K$-linear map satisfying the \emph{Leibniz rule}
\begin{equation}
  \label{eq:leibniz}
  \der(fg) =  \der(f) \, g + f \, \der(g).
\end{equation}
For convenience, we may assume $K \le \galg$, and we write $f'$ as a shorthand for $\der(f)$. The following definition~\cite{RosenkranzRegensburger2008a} captures the algebraic properties of the Fundamental Theorem of Calculus
and Integration by Parts.

\begin{definition}
\label{def:intdiffalg}
We call $(\galg, \der, \cum)$ an \emph{integro-differential algebra} if $(\galg,\der)$ is a commutative differential $K$-algebra and $\cum$ is a $K$-linear section (right inverse) of $\der$, i.e.
\begin{equation}
\label{eq:section-axiom}
    (\cum f)' = f,
\end{equation}
such that the \emph{differential Baxter axiom}
\begin{equation}
\label{eq:diff-baxter-axiom}
(\cum f')(\cum g') + \cum (f g)' = (\cum f') g + f (\cum g')
\end{equation}
 holds.
\end{definition}

We refer to $\der$ and $\cum$ respectively as the \emph{derivation}
and \emph{integral} of $\galg$ and to~\eqref{eq:section-axiom} as \emph{section axiom}.
Moreover, we call a section $\cum$ of $\der$ an integral for $\der$ if it satisfies~\eqref{eq:diff-baxter-axiom}.
For the similar notion of differential Rota-Baxter algebras, we refer to~\cite{GuoKeigher2008} but see also below.

Note that we have applied \emph{operator notation} for the integral; otherwise, for example,
the section axiom~\eqref{eq:section-axiom} would read $(\cum(f))'=f$, which is quite unusual at least for an analyst. We will
likewise often use operator notation for the derivation, so the Leibniz rule~\eqref{eq:leibniz} can also be written as $\der fg = (\der f) g + f (\der g)$. For the future we also introduce the following
convention for saving parentheses: Multiplication has precedence over
integration, so $\cum f \cum g$ is to be parsed as $\cum (f \cum g)$.

Let us also remark that Definition~\ref{def:intdiffalg} can be
generalized: First, no changes are needed for the \emph{noncommutative case} (meaning
$\galg$ is noncommutative). This would for example be an appropriate
setting for matrices with entries in $\galg = C^\infty[a,b]$,
providing an algebraic framework for the results on linear systems of
ODEs. Second, one may add a \emph{nonzero weight}
in the Leibniz axiom, thus incorporating also discrete models where $\der$
is the difference operator defined by $(\der f)_k = f_{k+1} - f_k$. The nice thing is that all other axioms remain unchanged.
For both generalizations confer also to~\cite{GuoKeigher2008}.

We study first some direct consequences of the section axiom~\eqref{eq:section-axiom}. For further details on linear left and right inverses, we refer for example to~\cite[p.~211]{Bourbaki1998} or to~\cite{Nashed1976} in the context of generalized inverses. We also introduce the following names for the \emph{projectors and modules} associated with a section of a derivation.

\begin{definition}
 \label{def:ini-eval}
  Let $(\galg, \der)$ be a differential $K$-algebra and $\cum$ a $K$-linear section
  of $\der$. Then we call the projectors
  \begin{equation*}
    \ini = \cum \circ \der \qquad\text{and}\qquad
    \evl = 1-\cum \circ \der
  \end{equation*}
  respectively the \emph{initialization} and the
  \emph{evaluation} of $\galg$. Moreover, we refer to
  \begin{equation*}
    \const = \Ker{\der} = \Ker{\ini} = \Img{\evl}
    \quad\text{and}\quad
    \init = \Img{\cum} = \Img{\ini} = \Ker{\evl}
  \end{equation*}
  as the submodules of respectively \emph{constant} and \emph{initialized functions}.
\end{definition}

Note that they are indeed projectors since $\ini \circ \ini = \cum
\circ (\der \circ \cum) \circ \der = \ini$
by~\eqref{eq:section-axiom}, which implies $\evl \circ \evl = 1 - \ini
- \ini + \ini \circ \ini = \evl$. As is well known~~\cite[p.~209]{Bourbaki1998}, every projector is
characterized by its kernel and image---they form a direct
decomposition of the module into two submodules, and every such
decomposition corresponds to a unique projector. We have therefore a \emph{canonical decomposition}
\begin{equation*}
  \galg = \const \dirs \init,
\end{equation*}
which allows to split off the ``constant part'' of every ``function'' in $\galg$.

Before turning to the other axioms, let us check what all this means
in the \emph{standard example} $\galg = C^\infty(\mathbb{R})$ with $\der = \tfrac{d}{dx}$ and $\cum = \int_a^x$. Obviously, the elements of $\const$ are then indeed the constant functions $f(x) = c$, while $\init$ consists of those functions that
satisfy the homogeneous initial condition $f(a) = 0$. This also
explains the terminology for the projectors: Here $\evl f =
f(a)$ evaluates $f$ at the initialization point $a$, and $\ini f = f -
f(a)$ enforces the initial condition. Note that in this example the evaluation $\evl$ is multiplicative; we will show below that this holds in any integro-differential algebra.

The Leibniz rule~\eqref{eq:leibniz} and the differential Baxter axiom~\eqref{eq:diff-baxter-axiom}
entail interesting properties of the two submodules $\const$ and $\init$. For understanding these, it is more economic to
forget for a moment about integro-differential algebras and turn to
the following general observation about \emph{projectors on an algebra}. We use again operator notation, giving precedence to multiplication over the linear operators.

\begin{lemma}
  \label{lem:proj-alg}
  Let $E$ and $J$ be projectors on a $K$-algebra with $E + J = 1$, set
\[
C = \Img{E} = \Ker{J}\quad\text{and}\quad I = \Ker{E} = \Img{J}.
\]
  Then the following  statements are equivalent:
  \begin{enumerate}
  \item The projector $E$ is multiplicative, meaning $Efg=(Ef)(Eg)$.
  \item The projector $J$ satisfies the identity $(Jf)(Jg) + Jfg = (Jf)g + f(Jg)$ .
  \item The submodule $C$ is a subalgebra and the submodule $I$ an ideal.
  \end{enumerate}
\end{lemma}
\begin{proof}
  \underline{1. $\Leftrightarrow$ 2.} Multiplicativity of $E=1-J$ just means
\[
  fg - Jfg = f g - (Jf)g - f(Jg) + (Jf)(Jg).
\]

  \noindent
  \underline{1. $\Rightarrow$ 3.} This follows immediately
  because $C$ is the image and $I$ the kernel of the algebra
  endomorphism $E$.

  \noindent
  \underline{3. $\Rightarrow$ 1.} Let $f, g$ be arbitrary. Since the
  given $K$-algebra is a direct sum of $C$ and $I$, we have $f = f_C +
  f_I$ and $g = g_C + g_I$ for $f_C = E f, g_C = E g\in C$ and $f_I
  = J f, g_I = Jg \in I$. Then
  \begin{equation*}
    E f g = E f_C  g_C + E f_C  g_I + E f_I  g_C  + E f_I  g_I
  \end{equation*}
  Since $I$ is an ideal, the last three summands vanish. Furthermore, $C$ is a
  subalgebra, so $f_C g_C \in C$. This implies $E f_C g_C = f_C
  \, g_C$ because $E$ is a projector onto $C$.
\qed
\end{proof}

This lemma is obviously applicable to integro-differential algebras
$\galg$ with the projectors $E = \evl$ and $J = \ini$ and with the
submodules $C = \const$ and $I = \init$ because the differential
Baxter axiom~\eqref{eq:diff-baxter-axiom} is exactly condition 2. From
now on, we will therefore refer to $\const$ as the \emph{algebra of
constant functions} and to $\init$ as the \emph{ideal of initialized
functions}. Moreover, we note that in any integro-differential algebra the evaluation
$\evl = 1-\cum \circ \der$ is multiplicative, meaning
\begin{equation}
\label{eq:mult}
\evl f g=(\evl f)(\evl g).
\end{equation}
Altogether we obtain now the following characterization
of integrals (note that the requirement that $\const$ be a subalgebra
already follows from the Leibniz axiom).

\begin{corollary}
  \label{cor:intdiffalg-ints}
  Let $(\galg, \der)$ be a differential algebra. Then a section $\cum$
  of $\der$ is an integral if and only if its evaluation $\evl = 1 - \cum \circ \der$ is multiplicative,
  and if and only if $\init = \Img{\cum}$ is an ideal.
\end{corollary}

Note that the ideal $\init$ corresponding to an integral is in general
\emph{not a differential ideal} of $\galg$. We can see this already in
the standard example $C^\infty[0,1]$, where $\init$ consists of all $f
\in C^\infty[0,1]$ with $f(0) = 0$. Obviously $\init$ is not
differentially closed since $x \in \init$ but $x' = 1 \not\in \init$.

The above corollary implies immediately that an integro-differential
algebra $\galg$ can \emph{never be a field} since then the only
possibilities for $\init$ would be $0$ and $\galg$. The former case is
excluded since it means that $\Ker{\der} = \galg$, contradicting the
surjectivity of $\der$. The latter case corresponds to $\Ker{\der} =
0$, which is not possible because $\der 1 = 0$.

\begin{corollary}
  \label{cor:intdiffalg-not-field}
  An integro-differential algebra is never a field.
\end{corollary}

In some sense, this observation ensures that all integro-differential
algebras are \emph{fairly complicated}. The next result points in the
same direction, excluding finite-dimensional algebras.

\begin{proposition}
  The iterated integrals $1, \cum 1, \cum\cum 1, \ldots$ are all
  linearly independent over $K$. In particular, every integro-differential algebra is
  infinite-dimensional.
\end{proposition}
\begin{proof}
  Let $(u_n)$ be the sequence of iterated integrals of $1$. We prove
  by induction on~$n$ that $u_0, u_1, \ldots, u_n$ are linearly
  independent. The base case~$n=0$ is trivial. For the induction step
  from~$n$ to~$n+1$, assume $c_0 u_0 + \cdots + c_{n+1} u_{n+1} = 0$.
  Applying~$\der^{n+1}$ yields $c_{n+1} = 0$. But by the induction
  hypothesis, we have then also $c_0 = \cdots = c_n = 0$. Hence $u_0,
  \ldots, u_{n+1}$ are linearly independent. \qed
\end{proof}

Let us now return to our discussion of the \emph{differential Baxter
  axiom}~\eqref{eq:diff-baxter-axiom}. We will offer an equivalent
description that is closer to analysis. It is more compact but
less symmetric. (In the noncommutative case one has to add the
opposite version---reversing all products---for obtaining
equivalence.)

\begin{proposition}
  \label{prop:diff-baxter-axiom}
  The differential Baxter axiom~\eqref{eq:diff-baxter-axiom} is
  equivalent to
  \begin{equation}
    \label{eq:int-by-parts}
    f \cum g = \cum fg + \cum f' \cum g,
  \end{equation}
  in the presence of the Leibniz axiom~\eqref{eq:leibniz} and the section axiom~\eqref{eq:section-axiom}.
\end{proposition}
\begin{proof}
  For proving~\eqref{eq:int-by-parts} note that since $\init$ is an
  ideal, $f \cum g$ is invariant under the projector $\ini$ and thus
  equal to $\cum (f \cum g)' = \cum f' \cum g + \cum fg$ by the
  Leibniz axiom~\eqref{eq:leibniz} and the section
  axiom~\eqref{eq:section-axiom}. Alternatively, one can also
  obtain~\eqref{eq:int-by-parts} from~\eqref{eq:diff-baxter-axiom} if
  one replaces $g$ by $\cum g$ in~\eqref{eq:diff-baxter-axiom}.
  Conversely, assuming~\eqref{eq:int-by-parts} we see that $\init$ is
  an ideal of $\galg$, so Corollary~\ref{cor:intdiffalg-ints} implies
  that $\cum$ satisfies the differential Baxter
  axiom~\eqref{eq:diff-baxter-axiom}. \qed
\end{proof}

For obvious reasons, we refer to~\eqref{eq:int-by-parts} as
\emph{integration by parts}. The usual formulation $\cum fG' = fG -
\cum f'G$ is only satisfied ``up to a constant'', or if one restricts
$G$ to $\Img{\cum}$. Substituting $G = \cum g$ then leads
to~\eqref{eq:int-by-parts}. But note that we have now a more algebraic
perspective on this well-known identity of Calculus: It tells us how
$\init$ is realized as an ideal of $\galg$.

Sometimes a variation of~\eqref{eq:int-by-parts} is useful. Applying $\cum$ to the Leibniz axiom~\eqref{eq:leibniz} and
using the fact that $\evl = 1 - \ini$ is multiplicative~\eqref{eq:mult}, we obtain
\begin{equation}
  \label{eq:int-by-parts-eval}
  \cum fg' = fg - \cum f'g - (\evl{f})(\evl{g}),
\end{equation}
which we call the \emph{evaluation variant} of integration by parts (a
form that is also used in Calculus).  Observe that,
we regain integration by parts~\eqref{eq:int-by-parts}
upon replacing $g$ by $\cum g$ in~\eqref{eq:int-by-parts-eval} since $\evl \cum g = 0$.

Note that in general one cannot extend a given differential algebra to an integro-differential algebra since the
latter requires a \emph{surjective derivation}. For example,
in $(K[x^2], x\der)$ the image of $\der$ does not contain $1$. As
another example (cf.~Sect.~\ref{sec:intdiffpol}), the algebra of differential polynomials $\galg =
K\{u\}$ does not admit an integral in the sense of
Definition~\ref{def:intdiffalg} since the image of $\der$ does
not contain $u$.

How can we isolate the \emph{integro part} of an
integro-differential algebra? The disadvantage (and also advantage!)
of the differential Baxter axiom~\eqref{eq:diff-baxter-axiom} is that
it entangles derivation and integral. So how can one express
``integration by parts'' without referring to the derivation?

\begin{definition}
\label{def:Rota-Baxter}
  Let $\galg$ be a $K$-algebra and $\cum$ a $K$-linear operation
  satisfying
  \begin{equation}
    \label{eq:baxter-axiom}
    (\cum f)(\cum g) = \cum f \cum g + \cum g \cum f.
  \end{equation}
  Then $(\galg, \cum)$ is called a \emph{Rota-Baxter algebra} (of weight zero).
\end{definition}

Rota-Baxter algebras are named after Glen Baxter~\cite{Baxter1960} and
Gian-Carlo Rota~\cite{Rota1969}; see also \cite{Guo2002,Guo2009} for further details.
In the following, we refer to~\eqref{eq:baxter-axiom} as \emph{Baxter
  axiom}; in contrast to the differential Baxter axiom~\eqref{eq:diff-baxter-axiom}, we will sometimes also call it the
\emph{pure Baxter axiom}.

One might now think that an integro-differential algebra $(\galg,
\der, \cum)$ is a differential algebra $(\galg, \der)$ combined with a
Rota-Baxter algebra $(\galg, \cum)$ such that the section
axiom~\eqref{eq:section-axiom} is satisfied. In fact, such a structure
was introduced independently by Guo and Keigher~\cite{GuoKeigher2008} under the name \emph{differential
  Rota-Baxter algebras}. But we will see that an integro-differential
algebra is a little bit more---this is why we also refer
to~\eqref{eq:baxter-axiom} as ``weak Baxter axiom'' and
to~\eqref{eq:diff-baxter-axiom} and~\eqref{eq:int-by-parts} as
``strong Baxter axioms''.

\begin{proposition}
  \label{prop:pure-baxter-alg}
  Let $(\galg, \der)$ be a differential algebra and $\cum$ a section
  for $\der$. Then $\cum$ satisfies the pure Baxter
  axiom~\eqref{eq:baxter-axiom} if and only if $\init = \Img{\cum}$ is a
  subalgebra of $\galg$. In particular, $(\galg, \cum)$ is a Rota-Baxter
  algebra for any integro-differential algebra $(\galg, \der, \cum)$.
\end{proposition}
\begin{proof}
  Clearly~\eqref{eq:baxter-axiom} implies that $\init$ is a subalgebra
  of~$\galg$. Conversely, if $(\cum f)(\cum g)$ is contained in
  $\init$, it is invariant under the projector $\ini$ and
  must therefore be equal to $\cum \der \; (\cum f)(\cum g) = \cum f \cum g +
  \cum g \cum f$ by the Leibniz axiom~\eqref{eq:leibniz}.
 \qed
\end{proof}

So the strong Baxter axiom~\eqref{eq:diff-baxter-axiom} requires that
$\init$ be an ideal, the weak Baxter axiom~\eqref{eq:baxter-axiom}
only that it be a subalgebra. We will soon give a counterexample for
making sure that~\eqref{eq:diff-baxter-axiom} is indeed asking for
more than~\eqref{eq:baxter-axiom}, see Example~\ref{ex:counter}.
But before this we want to express the difference between the two axioms in terms of a \emph{linearity
  property}. Recall that both $\der$ and $\cum$ were introduced as
$K$-linear operations on $\galg$. Using the Leibniz
axiom~\eqref{eq:leibniz}, one sees immediately that $\der$ is
even $\const$-linear. It is natural to expect the same from $\cum$,
but this is exactly the difference
between~\eqref{eq:diff-baxter-axiom} and~\eqref{eq:baxter-axiom}.

\begin{proposition}
  \label{prop:baxter-clinear}
  Let $(\galg, \der)$ be a differential algebra and $\cum$ a section
  for $\der$. Then $\cum$ satisfies the differential Baxter
  axiom~\eqref{eq:diff-baxter-axiom} if and only if it satisfies the pure Baxter
  axiom~\eqref{eq:baxter-axiom} and is $\const$-linear.
\end{proposition}
\begin{proof}
  Assume first that $\cum$ satisfies the differential Baxter
  axiom~\eqref{eq:diff-baxter-axiom}. Then the pure Baxter
  axiom~\eqref{eq:baxter-axiom} holds by
  Proposition~\ref{prop:pure-baxter-alg}. For proving $\cum cg = c \,
  \cum g$ for all $c \in \const$ and $g \in \galg$, we use the
  integration-by-parts formula~\eqref{eq:int-by-parts} and $c' = 0$.

  Conversely, assume the pure Baxter axiom~\eqref{eq:baxter-axiom} is
  satisfied and $\cum$ is $\const$-linear. By
  Proposition~\ref{prop:diff-baxter-axiom} it suffices to prove the
  integration-by-parts formula~\eqref{eq:int-by-parts} for $f, g \in
  \galg$. Since $\galg = \const \dirs \init$, we may first consider
  the case $f \in \const$ and then the case $f \in \init$. But the
  first case follows from $\const$-linearity; the second case means $f
  = \cum \tilde{f}$ for $\tilde{f} \in \galg$,
  and~\eqref{eq:int-by-parts} becomes the pure Baxter
  axiom~\eqref{eq:baxter-axiom} for $\tilde{f}$ and $g$.
  \qed
\end{proof}

Let us now look at some natural \emph{examples of integro-differential
  algebras}, in addition to our standard examples $C^\infty(\R)$ and $C^\infty[a,b]$.

\begin{rrexample}
  \label{ex:intdiffalg-analysis}
  The \emph{analytic functions} on the real interval $[a,b]$ form an
  integro-differential subalgebra $C^\omega[a,b]$ of $C^\infty[a,b]$
  over $K = \R$ or $K = \C$. It contains in turn the
  integro-differential algebra $\exppol$ of \emph{exponential
    polynomials}, defined as the space of all $K$-linear combinations
  of $x^n e^{\lambda x}$, with $n \in \N$ and $\lambda \in K$.
  Finally, the algebra of \emph{ordinary polynomials} $K[x]$ is an
  integro-differential subalgebra in all cases.
\end{rrexample}

All the three examples above have \emph{algebraic analogs},
with integro-differential structures defined in the expected way.

\begin{rrexample}
  For a field $K$ of characteristic zero, the \emph{formal power
    series} $K[[x]]$ are an integro-differential algebra. One sets
  $\der x^k = k x^{k-1}$ and $\cum x^k = x^{k+1}/(k+1)$; note that the
  latter needs characteristic zero. The formal power series contain a
  highly interesting and important integro-differential subalgebra:
  the \emph{holonomic power series}, defined as those whose
  derivatives span a finite-dimensional $K$-vector
  space~\cite{ChyzakSalvy1998,SalvyZimmerman1994}.

  Of course $K[[x]]$ also contains (an isomorphic copy of) the
  integro-differential algebra of \emph{exponential polynomials}. In
  fact, one can define $\exppol$ algebraically as a
  quotient of the free algebra generated by the symbols~$x^k$ and~$e^{\lambda x}$, with~$\lambda$ ranging over~$K$. Derivation and
  integration are then defined in the obvious way. The exponential
  polynomials contain the \emph{polynomial ring} $K[x]$ as an
  integro-differential subalgebra. When $K=\R$ or $K=\C$, we use the notation
  $K[x]$ and $\exppol$ both for the analytic and the algebraic
  object since they are isomorphic.
\end{rrexample}

The following example is a clever way of transferring the previous example
to coefficient fields of \emph{positive characteristic}.

\begin{rrexample}
  \label{ex:Hurwitz}
  Let $K$ be an arbitrary field (having zero or positive
  characteristic). Then the algebra $\Hur$ of \emph{Hurwitz
    series}~\cite{Keigher1997} over $K$ is defined as the $K$-vector
  space of infinite $K$-sequences with the multiplication defined as
  \begin{equation*}
    (a_n) \cdot (b_n) = \bigg(\sum_{i=0}^n \binom{n}{i} \, a_i
    b_{n-i}\bigg)_n
  \end{equation*}
  for all $(a_n), (b_n) \in \Hur$. If one introduces derivation and
  integration by
  \begin{align*}
    &\der \, (a_0, a_1, a_2, \dots) = (a_1, a_2, \dots),\\
    &\cum \, (a_0, a_1, \dots) = (0, a_0, a_1, \dots),
  \end{align*}
  the Hurwitz series form an integro-differential algebra $(\Hur,
  \der, \cum)$, as explained by~\cite{KeigherPritchard2000} and~%
  \cite{Guo2002}. Note that as an additive group, $\Hur$ coincides
  with the formal power series $K[[z]]$, but its multiplicative
  structure differs: We have an isomorphism
  \begin{equation*}
    \sum_{n=0}^\infty a_n \, z^n \:\mapsto\: (n! \, a_n)
  \end{equation*}
  from $K[[z]]$ to $\Hur$ if and only if $K$ has characteristic zero.
  The point is that one can integrate every element of $\Hur$, whereas
  the formal power series $z^{p-1}$ does not have an antiderivative
  in $K[[z]]$ if $K$ has characteristic $p>0$.
\end{rrexample}

Now for the promised \emph{counterexample} to the claim that the
section axiom would suffice for merging a differential algebra
$(\galg, \der)$ and a Rota-Baxter algebra $(\galg, \cum)$ into an
integro-differential algebra $(\galg, \der, \cum)$.

\begin{rrexample}
\label{ex:counter}
  \newcommand{\orcum}{\cuma\!}
  Set $R = K[y]/(y^4)$ for $K$ a field of characteristic zero and define
  $\der$ on $\galg = R[x]$ as usual. Then $(\galg, \der)$ is a
  differential algebra. Let us define a $K$-linear map $\cum$ on
  $\galg$ by
  \begin{equation}
    \label{eq:baxter-not-diff}
    \cum f = \orcum f + f(0,0) \, y^2,
  \end{equation}
  where $\orcum$ is the usual integral on $R[x]$ with $x^k \mapsto
  x^{k+1}/(k+1)$. Since the second term vanishes under $\der$, we see
  immediately that $\cum$ is a section of $\der$. For verifying the
  pure Baxter axiom~\eqref{eq:baxter-axiom}, we compute
  \begin{align*}
    &(\cum f)(\cum g) = (\orcum f)(\orcum g) + y^2 \, \orcum \, \big(
    g(0,0) \, f + f(0,0) \, g \big) + f(0, 0) \, g(0, 0) \, y^4,\\
    &\cum f \cum g = \cum f \big(\orcum g + g(0,0) \, y^2\big)
    = \orcum f \orcum g + g(0,0) \, y^2 \, \orcum f.
  \end{align*}
  Since $y^4 \equiv 0$ and the ordinary integral $\orcum$ fulfills the
  pure Baxter axiom~\eqref{eq:baxter-axiom}, this implies immediately
  that $\cum$ does also. However, it does not fulfill the differential
  Baxter axiom~\eqref{eq:diff-baxter-axiom} because it is not
  $\const$-linear: Observe that $\const$ is here $\Ker{\der} = R$, so
  in particular we should have $\cum (y \cdot 1) = y \cdot \cum 1$.
  But one checks immediately that the left-hand side yields $xy$,
  while the right-hand side yields $xy+y^3$.
\end{rrexample}

\subsection{Ordinary Integro-Differential Algebras.}

The following example shows that our current notion of integro-differential algebra
includes also algebras of ``multivariate functions''.

\begin{rrexample}
  \label{ex:partial-intdiffalg}
  Consider $\galg = C^\infty(\R^2)$ with the derivation $\partial u =
  u_x + u_y$. Finding sections for $\partial$ means solving the
  \emph{partial differential equation} $u_x + u_y = f$. Its general
  solution is given by
  \begin{equation*}
    u(x,y) = \int_a^x f(t,y-x+t)\,dt + g(y-x),
  \end{equation*}
  where $g \in C^\infty(\R)$ and $a \in \R$ are arbitrary. Let us
  choose $a=0$ for simplicity. In order to ensure a linear section,
  one has to choose $g = 0$, arriving at
  \begin{equation*}
    \cum f = \int_0^x f(t,y-x+t)\,dt,
  \end{equation*}
  Using a change of variables, one may verify that $\cum$ satisfies
  the pure Baxter axiom~\eqref{eq:baxter-axiom}, so $(\galg, \cum)$ is
  a Rota-Baxter algebra.

  We see that the \emph{constant functions} $\const = \Ker{\der}$ are
  given by $(x,y) \mapsto c(x-y)$ with arbitrary $c \in C^\infty(\R)$,
  while the \emph{initialized functions} $\init = \Img{\cum}$ are
  those $F \in \galg$ that satisfy $F(0, y) = 0$ for all $y \in \R$.
  In other words, $\const$ consists of all functions constant on the
  characteristic lines $x - y = \mathrm{const}$, and $\init$ of those
  satisfying the homogeneous initial condition on the vertical axis
  (which plays the role of a ``noncharacteristic initial manifold'').
  This is to be expected since $\cum$ integrates along the
  characteristic lines starting from the initial manifold. The
  \emph{evaluation} $\evl\colon \galg \rightarrow \galg$ maps a
  function $f$ to the function $(x,y) \mapsto f(0,y-x)$. This means
  that $f$ is ``sampled'' only on the initial manifold, effectively
  becoming a univariate function: the general point $(x,y)$ is
  projected along the characteristics to the initial point $(0,y-x)$.

  Since $\evl$ is multiplicative on $\galg$, Lemma~\ref{lem:proj-alg}
  tells us that $(\galg, \der, \cum)$ is in fact an
  \emph{integro-differential algebra}. Alternatively, note that
  $\init$ is an ideal and that $\cum$ is $\const$-linear. Furthermore,
  we observe that here the polynomials are given by $K[x]$.
\end{rrexample}

In the following, we want to restrict
ourselves to boundary problems for \emph{ordinary differential
  equations}. Hence we want to rule out cases like
Example~\ref{ex:partial-intdiffalg}. The most natural way for
distinguishing ordinary from partial differential operators is to look
at their kernels since only the former have finite-dimensional ones. Note that in the following definition we
deviate from the standard terminology in differential algebra~\cite[p.~58]{Kolchin1973}, where ordinary only refers to
having a single derivation.

From now on, we restrict the ground ring $K$ to a \emph{field}. We can now characterize when a differential algebra is ordinary by requiring that $\const$ be one-dimensional over $K$, meaning $\const = K$.

\begin{definition}
  \label{def:ord-intdiffalg}
  A differential algebra $(\galg, \der)$ is called \emph{ordinary} if $\dim_K{\Ker{\der}}
  = 1$.
\end{definition}

Note that except for Example~\ref{ex:partial-intdiffalg} all our examples have been
ordinary integro-differential algebras. The requirement of ordinariness has a number of pleasant consequences.
First of all, the somewhat tedious distinction between the weak and
strong \emph{Baxter axioms} disappears since now $\galg$ is an algebra over its own field of constants $K = \const$.
Hence $\cum$ is by definition $\const$-linear, and
Lemma~\ref{prop:baxter-clinear} ensures that the pure Baxer
axiom~\eqref{eq:baxter-axiom} is equivalent to the differential
Baxter axiom~\eqref{eq:diff-baxter-axiom}. Let us summarize this.

\begin{corollary}
  \label{cor:strong-weak-baxter}
  In an ordinary integro-differential algebra, the constant functions
  coincide with the ground field, and the strong and weak Baxter
  axioms are equivalent.
\end{corollary}

Recall that a \emph{character} on an algebra (or group) is a
multiplicative linear functional; this may be seen as a special case
of the notion of character in representation theory, namely the case
when the representation is one-dimensional. In our context, a character on an integro-differential
algebra $\galg$, is a $K$-linear map $\phi\colon \galg \rightarrow K$ satisfying $\phi(fg) =
\phi(f) \, \phi(g)$ and a fortiori also $\phi(1) = 1$. So we just require $\phi$ to be a $K$-algebra homomorphism, as for example in~\cite[p.~407]{Lang1993}.

Ordinary integro-differential algebras will always have at least one
character, namely the \emph{evaluation}: One knows from Linear Algebra
that a projector $P$ onto a one-dimensional subspace $[w]$ of a
$K$-vector space $V$ can be written as $P(v) = \phi(v) \, w$, where
$\phi\colon V \rightarrow K$ is the unique functional with $\phi(w) =
1$. If~$V$ is moreover a $K$-algebra, a projector onto $K = [1]$ is
canonically described by the functional $\phi$ with normalization
$\phi(1) = 1$. Hence multiplicative projectors like $\evl$ can be
viewed as characters. In the next section, we consider other
characters on $\galg$; for the moment let us note $\evl$ as a
distinguished character. We write $\multfunc$ for the set of all
nonzero characters on a $K$-algebra $\galg$, in other words all
algebra homomorphisms $\galg \rightarrow K$.

One calls a $K$-algebra \emph{augmented} if there exists a character
on it. Its kernel $\init$ is then known as an \emph{augmentation
  ideal} and forms a direct summand of $K$; see for
example~\cite[p.~132]{Cohn2003}. Augmentation ideals are always
maximal ideals (generalizing the $C^\infty[a,b]$ case) since the
direct sum $\galg = K \dirs \init$ induces a ring isomorphism $\galg
\!/\! \init \cong K$. Corollary~\ref{cor:intdiffalg-ints} immediately
translates to the following characterization of integrals in ordinary
differential algebras.

\begin{corollary}
\label{cor:intcharacter}
  In an ordinary differential algebra $(\galg, \der)$, a section
  $\cum$ of $\der$ is an integral if and only if its evaluation is a character
  if and only if $\init = \Img{\cum}$ is an augmentation ideal.
\end{corollary}

\subsection{Initial Value Problems.}

It is clear that in general we cannot assume that the solutions of a differential equation with coefficients in
$\galg$ are again in $\galg$. For example, in $\galg = K[x]$, the differential equation $u' - u = 0$ has no solution. In fact, its ``actual'' solution space is spanned by $u(x) = e^x$ if $K = \R$ or $K
= \C$. So in this case we should have taken the exponential
polynomials $\galg = \exppol$ for ensuring that $u \in \galg$. But if
this is the case, we can also solve the \emph{inhomogeneous
  differential equation} $u' - u = f$ whose general solution is $K e^x + e^x \cum e^{-x} f$, with $\cum = \cum_0^x$ as usual. Of course we can also incorporate the initial condition $u(0) = 0$, which leads
to $u = e^x \cum e^{-x} f$.

This observation generalizes: Whenever we can solve the
homogeneous differential equation within $\galg$, we can also solve
the initial value problem for the corresponding inhomogeneous problem.
The classical tool for achieving this explicitly is the
\emph{variation-of-constants
  formula}~\cite[p.~74]{CoddingtonLevinson1955}, whose abstract
formulation is given in Theorem~\ref{thm:ivp} below.

As usual~\cite{PutSinger2003}, we will write $\gdiffop$ for the
ring of differential operators with coefficients in $\galg$, see also Sect.~\ref{sec:intdiffop}. Let
\[
 T=\der^n  + c_{n-1} \der^{n-1} + \cdots + c_0
\]
be a monic (i.e. having leading coefficient~$1$) differential operator in $\gdiffop$ of degree $n$.
Then we call $u_1, \ldots, u_n \in \galg$ a \emph{fundamental system} for $T$ if it is a $K$-basis
for $\Ker{T}$, so it yields the right number of solutions for the
homogeneous differential equation $Tu=0$. A fundamental system will
be called \emph{regular} if its associated Wronskian matrix
\begin{equation*}
  W(u_1,\ldots,u_n) =
  \begin{pmatrix}
    u_1 & \cdots & u_n\\
    u_1' & \cdots & u_n'\\
    \vdots & \ddots & \vdots\\
    u_1^{(n-1)} & \cdots & u_n^{(n-1)}
  \end{pmatrix}
\end{equation*}
is invertible in $\galg^{n \times n}$ or equivalently if its Wronskian $\det W(u_1,\ldots,u_n)$
is invertible in $\galg$. Of course this alone implies already that $u_1, \ldots, u_n$ are linearly independent.

\begin{definition}
  \label{def:reg-diffop}
  A monic differential operator $T \in \gdiffop$ is called \emph{regular} if it has a regular fundamental system.
\end{definition}

For such differential operators, variation of constants goes
through---the canonical initial value problem can be solved uniquely.
This means in particular that regular differential operators are
always \emph{surjective}.

\begin{theorem}
  \label{thm:ivp}
  Let $(\galg, \der, \cum)$ be an ordinary integro-differential
  algebra. Given a regular differential operator $T \in \gdiffop$ with $\deg{T}=n$ and a
  regular fundamental system $u_1, \ldots, u_n \in \galg$, the canonical
  initial value problem
  \begin{equation}
    \label{eq:ivp}
    \bvp{Tu=f}{\evl u = \evl u' = \cdots = \evl u^{(n-1)} = 0}
  \end{equation}
  has the unique solution
  \begin{equation}
    \label{eq:inhom}
    u = \sum_{i=1}^n u_i \cum d^{-1}\,d_i f
  \end{equation}
  for every $f \in \galg$, where $d=\det W(u_1,\ldots,u_n)$, and $d_i$ is
  the determinant of the matrix $W_i$ obtained from $W$ by replacing
  the $i$-th column by the $n$-th unit vector.
\end{theorem}

%%%%%%%%%%%%%%%%%%%%%%%%%%%%\begin{optional}%%%%%%%%%%%%%%%%%%%%%%%%%%%%

\begin{rrproof}
  We can use the usual technique of reformulating $Tu = f$ as a system
  of linear first-order differential equations with companion
  matrix~$A \in \galg^{n \times n}$. We extend the action of the
  operators $\cum, \der, \evl$ componentwise to $\galg^n$. Setting now
  \begin{equation*}
    \hat{u} = W \!\cum W^{-1}\! \: \hat{f}
  \end{equation*}
  with $\hat{f} = \trp{(0, \dots, 0, f)} \in \galg^n$, we check that
  $\hat{u} \in \galg^n$ is a solution of the first-order system
  $\hat{u}' = A \hat{u} + \hat{f}$ with initial condition
  $\evl(\hat{u}) = 0$. Indeed we have $\hat{u}' = W' \cum W^{-1}
  \hat{f} + W W^{-1} \hat{f}$ by the Leibniz rule and $AW = W'$ since
  $u_1, \ldots, u_n$ are solutions of $Tu = 0$; so the differential
  system is verified. For checking the initial condition, note that
  $\evl \cum W^{-1} \hat{f}$ is already the zero vector, so we have
  also $\evl(\hat{u}) = 0$ since $\evl$ is multiplicative.

  Writing $u$ for the first component of $\hat{u}$, we obtain a
  solution of the initial value problem~\eqref{eq:ivp}, due to the
  construction of the companion matrix. Let us now compute $\hat{g} =
  W^{-1} \hat{f}$. Obviously $\hat{g}$ is the solution of the linear
  equation system $W \hat{g} = \hat{f}$. Hence Cramer's rule, which is
  also applicable for matrices over rings~\cite[p.~513]{Lang2002},
  yields $\hat{g}_i$ as $d^{-1} d_i f$ and hence
  \begin{equation*}
    u = (W \! \cum \hat{g})_1 = \sum_{i=1}^n u_i \cum d^{-1}\,d_i f
  \end{equation*}
  since the first row of $W$ is $(u_1,\ldots,u_n)$.

  For proving uniqueness, it suffices to show that the homogeneous
  initial value problem only has the trivial solution. So assume $u$
  solves~\eqref{eq:ivp} with $f=0$ and choose coefficients $c_1,
  \ldots, c_n \in K$ such that
  \begin{equation*}
    u = c_1 u_1 + \cdots + c_n u_n.
  \end{equation*}
  Then the initial conditions yield $\evl (Wc) = 0$ with $c =
  \trp{(c_1, \dots, c_n)} \in K^n$. But we have also $\evl (Wc) =
  (\evl W) c$ because $\evl$ is linear, and $\det{\evl W} =
  \evl(\det{W})$ because it is moreover multiplicative. Since $\det{W}
  \in \galg$ is invertible, $\evl W \in K^{n \times n}$ is regular,
  so $c = {(\evl W)}^{-1} 0 = 0$ and $u = 0$.
\end{rrproof}

%%%%%%%%%%%%%%%%%%%%%%%%%%%%\end{optional}%%%%%%%%%%%%%%%%%%%%%%%%%%%%

\section{Integro-Differential Operators}
\label{sec:intdiffop}

With integro-differential algebras, we have algebraized the \emph{functions} to be used in differential equations and boundary problems, but we must also algebraize the \emph{operators} inherent in both---the differential operators on the left-hand side of the former, and the integral operators constituting the solution of the latter. As the name suggests, the integro-differential operators provide a data structure that contains both of these operator species. In addition, it has as a third species the boundary operators needed for describing (global as well as local) boundary conditions of any given boundary problem for a LODE.

\subsection{Definition.}

The basic idea is similar to the construction of the algebra of \emph{differential operators} $\gdiffop$ for a given differential algebra $(\galg, \der)$. But we are now starting from an ordinary integro-differential algebra $(\galg, \der, \cum)$, and the resulting algebra of integro-differential operators will accordingly be denoted by $\intdiffop$. Recall that $\gdiffop$ can be seen as the quotient of the free algebra generated by $\der$ and $f \in \galg$, modulo the ideal generated by the Leibniz rule $\der f = f \der + f'$. For $\intdiffop$, we do the same but with more generators and more relations. In the following, all integro-differential algebras are assumed to be ordinary.

Apart from~$\cum$, we will also allow a collection of ``point evaluations'' as new generators since they are needed for the specification of boundary problems. For example, the local boundary condition $u(1) = 0$ on a function $u \in \galg = C^\infty[0,1]$ gives rise to the functional $\evl_1 \in \galg^*$ defined by $u \mapsto u(1)$. As one sees immediately, $\evl_1$ is a \emph{character} on $\galg$, meaning $\evl_1(uv) = \evl_1(u) \, \evl_1(v)$ for all $u,v \in \galg$. This observation is the key for algebraizing ``point evaluations'' to an arbitrary integro-differential algebra where one cannot evaluate elements as in $C^\infty[0,1]$. We will see later how the characters serve as the basic building blocks for general local conditions like $3u(\pi)-2u(0)$ or global ones like $\cum_0^1 \xi u(\xi) \, d\xi$. Recall that we write $\multfunc$ for the set of all characters on integro-differential algebra $\galg$. In Sect.~\ref{sec:intdiffalg} we have seen that every  integro-differential algebra $(\galg, \der, \cum)$ contains at least one character, namely the evaluation $\evl = 1 - \cum\der$ associated with the integral. Depending on the application, one may add other characters.

\begin{definition}
  \label{def:intdiffop}
	  Let $(\galg, \der, \cum)$ be an ordinary integro-differential
	  algebra over a field~$K$ and $\Phi \subseteq \multfunc$. The
	  \emph{integro-differential operators} $\intdiffopchar{\Phi}$ are
	  defined as the free $K$-algebra generated by $\der$, and
	  $\cum$, the ``functions'' $f \in \galg$, and the characters
	  $\phi \in \Phi \cup \{ \evl \}$, modulo the rewrite rules in
	  Table~\ref{fig:red-rules}. If $\Phi$ is understood, we write
	  $\intdiffop$.
\end{definition}

\vskip-1.25em
\begin{table}[h]
  \newcommand{\ra}{\rightarrow}
  \centering
  \renewcommand{\baselinestretch}{1.3}
  \small
  \begin{tabular}[h]{|@{\quad}lcl@{\quad}|@{\quad}lcl@{\quad}|@{\quad}lcl@{\quad}|}
    \hline
    $f g$ & $\ra$ & $f \act g$ &
      $\der f$ & $\ra$ & $f \der + \der \act f$ &
      $\cum f \cum$ & $\ra$ & $(\cum \act f) \, \cum - \cum \, (\cum
        \act f)$\\
    $\phi \psi$ & $\ra$ & $\psi$ &
      $\der \phi$ & $\ra$ & $0$ &
      $\cum f \der$ & $\ra$ & $f - \cum (\der \act f) - (\evl \act f)
         \, \evl$\\
    $\phi f$ & $\ra$ & $(\phi \act f) \, \phi$ &
      $\der\cum$ & $\ra$ & $1$ &
      $\cum f \phi$ & $\ra$ & $(\cum \act f) \, \phi$\\
    \hline
  \end{tabular}
  \medskip
  \caption{Rewrite Rules for Integro-Differential Operators}
  \label{fig:red-rules}
\end{table}

The notation $U \cdot f$, used in the right-hand side of some of the
rules above, refers to the \emph{action} of $U \in \free$ on a
function $f \in \galg$; in particular, $f \cdot g$ denotes the product
of two functions $f, g \in \galg$. It is an easy matter to check that
the rewrite rules of Table~\ref{fig:red-rules} are fulfilled in
$(\galg, \der, \cum)$, so we may regard $\act$ as an action of
$\intdiffop$ on $\galg$. Thus every element $T \in \intdiffop$ acts as
a map $T\colon \galg \rightarrow \galg$.

We have given the relations as a \emph{rewrite system}, but their
algebraic meaning is also clear: If in the free algebra $\free$ of
Definition~\ref{def:intdiffop} we form the two-sided ideal $\green$
generated by the left-hand side minus right-hand side for each rule,
then $\intdiffopchar{\Phi} = \free/\green$. Note that there are
infinitely many such rules since each choice of $f,g \in \galg$ and
$\phi, \psi \in \Phi$ yields a different instance (there may be just
finitely many characters in $\Phi$ but the coefficient algebra $\galg$
is always infinite), so $\green$ is an infinitely generated ideal (it
was called the ``Green's ideal'' in~\cite{Rosenkranz2005} in a
slightly more special setting). Note that one gets back the rewrite
system of Table~\ref{fig:red-rules} if one uses the implied set of
generators and a suitable ordering (see Sect.~\ref{sec:rwtogb}).

The reason for specifying~$\green$ via a rewrite system is of course
that we may use it for generating a canonical simplifier
for~$\intdiffop$. This can be seen either from the term rewriting or
from the \grb\ basis perspective: In the former case, we see
Table~\ref{fig:red-rules} as a confluent and terminating rewrite
system (modulo the ring axioms); in the latter case, as a
\emph{noncommutative \grb\ basis with noetherian reduction} (its
elements are of course the left-hand side minus right-hand side for
each rule). While we cannot give a detailed account of these issues
here, we will briefly outline the \grb\ basis setting since our new
proof in Sect.~\ref{sec:rwtogb} will rely on it.

\subsection{Noncommutative \grb\ Bases.}

As detailed in Section~\ref{sec:theorema}, it is necessary for our application to deal  with  \emph{infinitely generated ideals} and an \emph{arbitrary set of indeterminates}. The following description of such a noncommutative \grb\ basis setting is based on the somewhat dated but still highly readable Bergman paper~\cite{Bergman1978}; for a summary
see~\cite[\S3.3]{Bueso}. For other approaches we refer the reader to
\cite{Mora1986,Mora1994,Ufnarovski1998,Ufnarovskij1995}.

Let us first recall some notions for abstract reduction relations~\cite{BaaderNipkow1998}. We consider a relation~$\mathord{\steprel} \subseteq A \times A$ for a set~$A$; typically~$\mathord{\steprel}$ realizes a single step in a simplification process like the transformation of integro-differential operators according to Table~\ref{fig:red-rules}. The transitive closure of~$\mathord{\steprel}$ is denoted by~$\mathord{\sredrel}$, its reflexive-transitive closure by~$\mathord{\redrel}$. We call $a \in A$ \emph{irreducible} if there is no $a_0 \in A$ with $a \steprel a_0$; we write $A_\downarrow$ for the set of all irreducible elements. If $a \redrel a_0$ with $a_0 \in A_\downarrow$, we call $a_0$ a \emph{normal form} of $a$, denoted by $\nf{a} = a_0$ in case it is unique.

If all elements are to have a unique normal form, we have to impose two conditions: termination for banning infinite reductions and confluence reuniting forks. More precisely, $\mathord{\steprel}$ is called \emph{terminating} if there are no infinite chains $a_1 \steprel a_2 \steprel \ldots$ and \emph{confluent} if for all $a, a_1, a_2 \in A$ the fork $a_1 \credrel a \redrel a_2$ finds a reunion $a_1 \redrel a_0 \credrel a_2$ for some $a_0 \in A$. If $\mathord{\steprel}$ is both terminating and confluent, it is called \emph{convergent}.

Turning to noncommutative Gr\"obner bases theory, we focus on reduction relations on the free algebra $\freealg$ over a commutative ring $K$ in an arbitrary set of indeterminates $X$; the corresponding monomials form the free monoid $\freemon$. Then a \emph{reduction system} for $\freealg$ is a set $\Sigma \subseteq \freemon \times \freealg$ whose elements are called rules. For a rule $\sigma = (W, f)$ and monomials $A,B \in \freemon$, the $K$-module endomorphism of $\freealg$ that fixes all elements of $\freemon$ except sending $AW\!B$ to $AfB$ is denoted by $\red{A}{\sigma}{B}$ and called a reduction. It is said to act trivially on $a \in \freealg$ if the coefficient of $AW\!B$ in $a$ is zero.

Every reduction system $\Sigma$ induces the relation $\mathord{\steprel}_\Sigma \subseteq \freealg \times \freealg$ defined by setting~$a \steprel_\Sigma b$ if and only if~$r(a) = b$ for some reduction acting nontrivially on~$a$. We call its reflexive-transitive closure $\redrel_\Sigma$ the \emph{reduction relation} induced by $\Sigma$, and we say that $a$ reduces to $b$ when $a \redrel_\Sigma b$. Accordingly we can speak of irreducible elements, normal forms, termination and confluence of $\Sigma$.

For ensuring termination, one can impose a noetherian \emph{monoid
  ordering} on $\freemon$, meaning a partial ordering such that $1 <
A$ for all $A \in \freemon$ and such that $B < B'$ implies $ABC <
AB'C$ for $A, B, B', C \in \freemon$. Recall that for partial
(i.e. not necessarily total) orderings, noetherianity means that there
are no infinite descending chains or equivalently that every nonempty
set has a minimal element~\cite[p.~156]{BeckerWeispfenning1993}. Note
that in a noetherian monoid ordering (like the divisibility relation
on natural numbers), elements are not always comparable.

Now if one has a noetherian monoid ordering on $\freemon$, then
$\Sigma$ will be \emph{terminating} provided it respects $<$ in the
sense that $W' < W$ for every rule $(W, f) \in \Sigma$ and every
nonzero monomial $W'$ of $f$. (Let us also remark that the condition
$1 < A$ from above might as well be dropped, as in~\cite{Bergman1978}:
The given rewrite system cannot contain a rule $1 \rightarrow f$ since
then $W < 1$ for at least one nonzero monomial~$W$ of~$f$, so $1 > W >
WW > \cdots$ would yield an infinite descending chain. Such rules
precluded, it is not stringent that constants in~$K$ be comparable
with the elements in~$X$. But since it is nevertheless very natural
and not at all restrictive, we stick to the monoid orderings as given
above.)

It is typically more difficult to ensure confluence of a reduction system $\Sigma$. According to the definition, we would have to investigate all forks $a_1 \credrel a \redrel a_2$, which are usually infinite in number. The key idea for a practically useful criterion is to consider only certain \emph{minimal forks} (called ambiguities below, following Bergman's terminology) and see whether their difference eventually becomes zero. This was first described by Buchberger in~\cite{Buchberger1965} for the commutative case; see also~\cite{Buchberger1970,Buchberger1998}. The general intuition behind minimal forks is analyzed in~\cite{Buchberger1987}, where \grb\ bases are compared with Knuth-Bendix completion and Robinson's resolution principle.

An \emph{overlap ambiguity} of $\Sigma$ is given by a quintuple $(\sigma, \tau, A, B, C)$ with $\Sigma$-rules $\sigma = (W, f)$, $\tau = (V, g)$ and monomials $A, B, C \in \freemon\setminus\{1\}$ such that $W = AB$ and $V = BC$. Its associated S-polynomial is defined as $fC-Ag$, and the ambiguity is called resolvable if the S-polynomial reduces to zero. (In general one may also have so-called inclusion ambiguities, but it turns out that one can always remove them without changing the resulting normal forms~\cite[\S5.1]{Bergman1978}. Since the reduction system of Table~\ref{fig:red-rules} does not have inclusion ambiguities, we will not discuss them here.)

For making the connection to ideal theory, we observe that every reduction system $\Sigma$ gives rise to a two-sided \emph{ideal} $I_\Sigma$ generated by all elements $W-f$ for $(W,f) \in \Sigma$; we have already seen this connection for the special case of the integro-differential operators. Note that $a \redrel_\Sigma 0$ is equivalent to $a \in I_\Sigma$.

In the given setting, the task of proving convergence can then be attacked by the so-called \emph{Diamond Lemma for Ring Theory}, presented as Theorem~1.2 in Bergman's homonymous paper~\cite{Bergman1978}; see also Theorem~3.21 in~\cite{Bueso}. It is the noncommutative analog of Buchberger's criterion~\cite{Buchberger1970} for infinitely generated ideals. (In the version given below, we have omitted a fourth equivalent condition that is irrelevant for our present purposes.)

\begin{theorem}
  \label{thm:diamond-lemma}
  Let $\Sigma$ be a reduction system for $\freealg$ and $\le$ a noetherian monoid ordering that respects $\Sigma$. Then the
  following conditions are equivalent:
  \begin{itemize}
  \item All ambiguities of $\Sigma$ are resolvable.
  \item The reduction relation $\redrel_\Sigma$ is convergent.
  \item We have the direct decomposition $\freealg = \irrmod \dirs
    I_\Sigma$ as $K$-modules.
  \end{itemize}
  When these conditions hold, the quotient algebra $\freealg/I_\Sigma$ may be identified with the $K$-module $\irrmod$, having the multiplication $a \cdot b = \nf{ab}$.
\end{theorem}

We will apply Theorem~\ref{thm:diamond-lemma} in Sect.~\ref{sec:rwtogb} for proving that Table~\ref{fig:red-rules} constitutes a \grb\ basis for the ideal $\green$. Hence we may conclude that $\intdiffop$ can be identified with the algebra $\free_\downarrow$ of \emph{normal forms}, and this is what gives us an algorithmic handle on the integro-differential operators. It is thus worth investigating these normal forms in some more detail.

\subsection{Normal Forms.}

We start by describing a set of \emph{generators}, which will subsequently be narrowed to the normal forms of $\intdiffopchar{\Phi}$. The key observation is that in any monomial we never need more than one integration while all the derivatives can be collected at its end.

\begin{lemma}
  \label{lem:pre_normal_form}
  Every integro-differential operator in $\intdiffopchar{\Phi}$ can be reduced to a linear combination of monomials $f \phi \cum g \psi \der^i$, where $i \ge 0$ and each of $f, \phi, \cum, g, \psi$ may also be absent.
\end{lemma}
\begin{rrproof}
  Call a monomial consisting only of functions and functionals ``algebraic''. Using the left column of Table~\ref{fig:red-rules}, it is immediately clear that all such monomials can be reduced to $f$ or $\phi$ or $f \phi$. Now let $w$ be an arbitrary monomial in the generators of $\intdiffopchar{\Phi}$. By using the middle column of Table~\ref{fig:red-rules}, we may assume that all occurrences of $\der$ are moved to the right, so that all monomials have the form $w = w_1 \cdots w_n \der^i$ with $i \ge 0$ and each of $w_1, \ldots, w_n$ either a function, a functional or $\cum$. We may further assume that there is at most one occurrence of $\cum$ among the $w_1, \ldots, w_n$. Otherwise the monomials $w_1 \cdots w_n$ contain $\cum \tilde{w} \cum$, where each $\tilde{w} = f \phi$ is an algebraic monomial. But then we can reduce
  \begin{displaymath}
    \cum \tilde{w} \cum = (\cum f \phi) \cum = (\cum \cdot f) \phi
    \cum
  \end{displaymath}
  by using the corresponding rule of Table~\ref{fig:red-rules}. Applying these rules repeatedly, we arrive at algebraic monomials left and right of $\cum$ (or just a single algebraic monomial if $\cum$ is absent).
\end{rrproof}

In \tma, the \emph{integro-differential operators} over an integro-differential algebra $\galg$ of coefficient functions are built up by $\mathtt{FreeIntDiffOp[\mathcal{F},K]}$. This functor constructs an instance of the monoid algebra with the word monoid over the infinite alphabet consisting of the letters $\der$ and $\cum$ along with a basis of $\galg$ and with all multiplicative characters induced by evaluations at points in $\mathtt{K}$:

\begin{mmacode}
\includegraphics[scale=0.75]{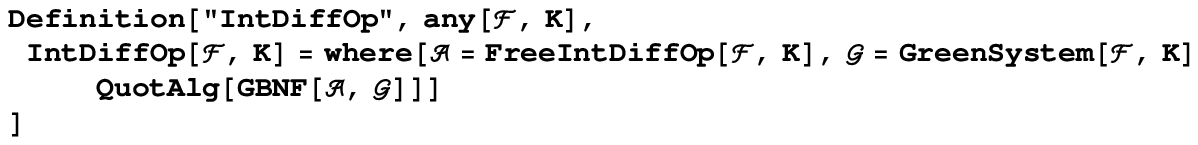}
\end{mmacode}

\noindent In this code fragment, the $\mathtt{GreenSystem}$ functor contains the encoding of the aforementioned rewrite system (Table~1), here understood as a noncommutative \grb\ basis. Normal forms for total reduction modulo infinite \grb\ bases are created by the $\mathtt{GBNF}$ functor, while the $\mathtt{QuotAlg}$ functor constructs the quotient algebra from the corresponding canonical simplifier (see Sect.~\ref{sec:theorema} for details). For instance, multiplying the integral operator $\cum$ by itself triggers an application of the Baxter rule:
\begin{mmacode}
\includegraphics[scale=0.75]{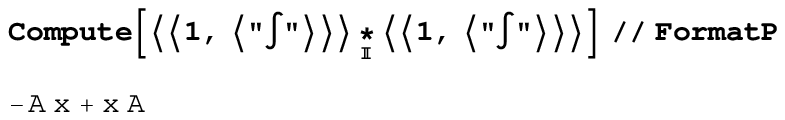}
\end{mmacode}
Here integral operators are denoted by $\mathtt{A}$, following the notation in the older implementation~\cite{Rosenkranz2005}.

We turn now to the normal forms of \emph{boundary conditions}. Since they are intended to induce mappings $\galg \rightarrow K$, it is natural to define them as those integro-differential operators that ``end'' in a character $\phi \in \Phi$. For example, if $\phi$ is the point evaluation $\evl_1$ considered before, the composition~$\evl_1 \der$ describes the local condition $u'(1) = 0$, the composition $\evl_1 \cum$ the global condition $\smash{\cum_0^1} u(\xi) \, d\xi = 0$. In general, boundary conditions may be arbitrary linear combinations of such composites; they are known as ``Stieltjes conditions'' in the literature~\cite{BrownKrall1974,BrownKrall1977}.

\begin{definition}
  \label{def:stieltjes}
  The elements of the right ideal
  \begin{displaymath}
    \scond = \Phi \cdot \intdiffopchar{\Phi}
  \end{displaymath}
  are called \emph{boundary conditions} over $\galg$.
\end{definition}

It turns out that their \emph{normal forms} are exactly the linear combinations of local and global conditions as in the example mentioned above. As a typical representative over $\galg = C^\infty(\R)$, one may think of an element like
\begin{equation*}
  \evl_0\der^2 + 3 \, \evl_\pi - 2 \, \evl_{2\pi} \cum \sin{x},
\end{equation*}
written as $u''(0) + 3 \, u(\pi) - 2 \cum_0^{2\pi} \sin{\xi} u(\xi) \, d\xi$ in traditional notation.

\begin{proposition}
  \label{prop:stieltjes}
  Every boundary condition of $\scond$ has the normal form
  \begin{equation}
    \label{eq:stieltjes}
    \sum_{\phi \in \Phi} \left( \sum_{i \in \N}
      a_{\phi,i} \, \phi \der^i + \phi \cum f_{\phi} \right),
  \end{equation}
 with only finitely many $a_{\phi,i} \in K$ and $f_{\phi} \in \galg$ nonzero.
\end{proposition}
\begin{rrproof}
  By Lemma~\ref{lem:pre_normal_form}, every boundary condition of
  $\scond$ is a linear combination of monomials having the form
  \begin{equation}
    \label{eq:mon_form}
    w = \chi f \phi \cum g \psi \der^i
    \quad\text{or}\quad
    w = \chi f \phi \der^i
  \end{equation}
  where each of $f, g, \phi, \psi$ may also be missing. Using the left column of Table~\ref{fig:red-rules}, the prefix $\chi f \phi$ can be reduced to a scalar multiple of a functional, so we may as well assume that $f$ and $\phi$ are not present; this finishes the right-hand case of~\eqref{eq:mon_form}. For the remaining case $w = \chi \cum g \psi \der^i$, assume first that $\psi$ is present. Then we have
  \begin{displaymath}
    \chi \, (\cum g \psi) = \chi \, (\cum \cdot g) \, \psi
    = (\chi \cum \cdot g) \, \chi \psi = (\chi \cum \cdot
    g) \, \psi,
  \end{displaymath}
  so $w$ is again a scalar multiple of $\psi \der^i$, and we are done. Finally, assume we have $w = \chi \cum g \der^i$. If $i = 0$, this is already a normal form. Otherwise we obtain
  \begin{displaymath}
    w = \chi \, (\cum g \der) \, \der^{i-1} = (\chi \cdot g) \,
    \chi \der^{i-1} - \chi \cum g' \der^{i-1} - (\evl \cdot g) \, \evl
    \der^{i-1},
  \end{displaymath}
  where the first and the last summand are in the required normal form, while the middle summand is to be reduced recursively, eventually leading to a middle term in normal form $\pm \chi \cum g' \der^0 = \pm \chi \cum g'$.
\end{rrproof}

Most expositions of boundary problems---both analytic and numeric ones---restrict their attention to local conditions, even more specifically to those with just two point evaluation (so-called two-point boundary problems). While this is doubtless the most important case, there are at least three reasons for considering \emph{Stieltjes boundary conditions} of the form~\eqref{eq:stieltjes}.
\begin{itemize}
	\item They arise in certain applications (e.g. heat radiated through a boundary) and in treating ill-posed problems by generalized Green's functions~\cite[p.~191]{Rosenkranz2005}.
	\item As we shall see (Sect.~\ref{sec:bp}), they are needed for factoring boundary problems.
	\item	Their algebraic description as a right ideal is very natural.
\end{itemize}
Hence we shall always mean all of $\scond$ when we speak of boundary conditions.

Let us now turn to the other two ingredients of integro-differential operators: We have already mentioned the \emph{differential operators} $\gdiffop$, but we can now see them as a subalgebra of $\intdiffopchar{\Phi}$. They have the usual normal forms since the Leibniz rule is part of the rewrite system. Analogously, we introduce the subalgebra of \emph{integral operators} generated by the functions and $\cum$. Using Lemma~\ref{lem:pre_normal_form}, it is clear that the normal forms of integral operators are $\galg$ itself and linear combinations of $f \cum g$, and the only rule applicable to them is $\cum f \cum \rightarrow \cdots$ in Table~\ref{fig:red-rules}. Since we have already included $\galg$ in $\gdiffop$, we introduce $\gintop$ as the $\galg$-bimodule generated by $\cum$ so that it contains only the monomials of the form $f \cum g$.

Finally, we must consider the two-sided ideal $\bndop$ of $\intdiffopchar{\Phi}$ generated by $\Phi$ whose elements are called \emph{boundary operators}. A more economical description of $\bndop$ is as the left $\galg$-submodule generated by $\scond$ because by Lemma~\ref{lem:pre_normal_form} any $w \chi \tilde{w}$ with $w, \tilde{w} \in \intdiffop$ can be reduced to $f \phi \cum g \psi \der^i \chi \tilde{w}$. Note that $\bndop$ includes all finite dimensional projectors $P$ along Stieltjes boundary conditions. Any such projector can be described in the following way: If $u_1, \dots, u_n \in \galg$ and $\beta_1, \dots, \beta_n \in \scond$ are biorthogonal (meaning $\beta_i(u_j) = \delta_{ij}$), then
\begin{equation}
  \label{eq:projector}
 P = \sum_{i=1}^n u_i \, \beta_i : \: \galg \rightarrow \galg
\end{equation}
is the projector onto $[u_1, \dots, u_n]$ whose kernel is the subspace of all $u \in \galg$ such that
 $\be(u)=0$ for all $ \be \in [\beta_1,\ldots,\beta_n]$. See for example~\cite[p.~71]{Kothe1969} or~\cite{RegensburgerRosenkranz2009} for details. Note that all elements of $\bndop$ have the normal form~\eqref{eq:projector}, except that the $(u_j)$ need not be biorthogonal to the $(\beta_i)$.

We can now characterize the normal forms of~$\intdiffopchar{\Phi}$ in a very straightforward and intuitive manner: Every monomial is either a \emph{differential operator} or an \emph{integral operator} or a \emph{boundary operator}. Hence every element of $\intdiffopchar{\Phi}$ can be written uniquely as a sum $T + G + B$, with a differential operator $T \in \gdiffop$, an integral operator $G \in \gintop$, and a boundary operator $B \in \bndop$.

\begin{proposition}
  \label{prop:intdiffop_normal}
  For an integro-differential algebra $\galg$ and characters $\Phi \subseteq \multfunc$, we have the direct decomposition $\intdiffopchar{\Phi} = \gdiffop \dirs \gintop \dirs \bndop$.
\end{proposition}
\begin{rrproof}
  Inspection of Table~\ref{fig:red-rules} confirms that all
  integro-differential operators having the described sum
  representation $T + G + P$ are indeed in normal form. Let us now
  prove that every integro-differential operator of
  $\intdiffopchar{\Phi}$ has such a representation. It is sufficient
  to consider its monomials $w$. If $w$ starts with a functional, we
  obtain a boundary condition by Proposition~\ref{prop:stieltjes}; so
  assume this is not the case. From Lemma~\ref{lem:pre_normal_form} we
  know that
  \begin{displaymath}
    w = f \phi \cum g \psi \der^i
    \quad\text{or}\quad
    w = f \phi \der^i,
  \end{displaymath}
  where each of $\phi, g, \psi$ may be absent. But $w \in \bndop$
  unless $\phi$ is absent, so we may actually assume
  \begin{displaymath}
    w = f \cum g \psi \der^i
    \quad\text{or}\quad
    w = f \der^i.
  \end{displaymath}
  The right-hand case yields $w \in \gdiffop$. If $\psi$ is present in
  the other case, we may reduce $\cum g \psi$ to $(\cum \cdot g) \,
  \psi$, and we obtain again $w \in \bndop$. Hence we are left with
  $w = f \cum g \der^i$, and we may assume $i > 0$ since otherwise we
  have $w \in \gintop$ immediately. But then we can reduce
  \begin{align*}
    w &= f \, (\cum g \der) \, \der^{i-1} = f \Big(g - \cum \, (\der
    \cdot g) - (\evl \cdot
    g) \, \evl \Big) \der^{i-1}\\
    &= (f g) \, \der^{i-1} - f \, \cum \, (\der \cdot g) \, \der^{i-1}
    - (\evl \cdot g) \, f \, \evl \, \der^{i-1},
  \end{align*}
  where the first term is obviously in $\gdiffop$ and the last one in
  $\bndop$. The middle term may be reduced recursively until the
  exponent of $\der$ has dropped to zero, leading to a term in
  $\gintop$.
\end{rrproof}

We can observe the direct decomposition $\intdiffopchar{\Phi} = \gdiffop \dirs \gintop \dirs \bndop$ in the following \emph{sample multiplication} of $\cum \der$ and $\der \der x e^x \cum$:
\begin{mmacode}
\includegraphics[scale=0.75]{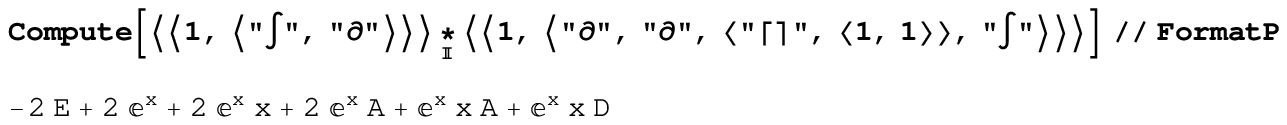}
\end{mmacode}
As in the previous computation, $\mathtt{A}$ stands for the integral $\cum$, moreover $\mathtt{D}$ for the derivation $\der$, and $\mathtt{E}$ for the evaluation. As we can see, the sum is composed of one differential operator (the last summand), two integral operators (in the middle), and three boundary operators (the first summands). Observe also that the input operators are not in normal form but the output operator is.

\subsection{Basis Expansion.}

Regarding the canonical forms for $\intdiffop$, there is one more issue that we have so far swept under the rug. The problem is that in the current setup elements like $x \der + 3x^2 \der$ and $(x + 3x^2) \der$ are considered distinct normal forms. More generally, if $f + g = h$ holds in $\galg$, there is no rule that allows us to recognize that $f + g \in \intdiffop$ and $h \in \intdiffop$ are the same. Analogously, if $\lambda \! \tilde{f} = \tilde{g}$ holds in $\galg$ with $\lambda\in K$, then $\lambda \! \tilde{f}$ and $\tilde{g}$ are still considered to be different in $\intdiffop$. A slightly less trivial example is when $f = (\cos{x})(\cos^2{x^2})$ and $g = -(\sin{x})(\sin{x^2})$ so that $h = \cos{(x+x^2)}$. What is needed in general is obviously a \emph{choice of basis} for $\galg$. But since such a choice is always to some degree arbitrary, we would like to postpone it as much as possible.

An unbiased way of introducing all $K$-linear relations in $\galg$ is simply to collect them in all in the two-sided ideal
\begin{equation*}
  \basis = ( f + g - h, \, \lambda \! \tilde{f} - \tilde{g} \mid
     f + g = h \text{ and } \lambda \! \tilde{f} = \tilde{g} \text{ in } \galg ) \subseteq \free,
\end{equation*}
which we shall call the \emph{linear ideal}. Since $\basis + \green$ corresponds to a unique ideal $\tilde{\basis}$ in $\intdiffop$, the necessary refinement of $\intdiffop$ can now be defined as
\begin{equation*}
\intdiffopbas = \intdiffop/\tilde{\basis} \cong \free/(\basis+\green)
\end{equation*}
whose elements shall be called \emph{expanded integro-differential operators}. Note that $\tilde{\basis}$ is really the ``same'' ideal as $\basis$ except that now $f,g,h,\tilde{f},\tilde{g} \in \intdiffop$. By the isomorphism above, coming from the Third Isomorphism Theorem~\cite[Thm 1.22]{Cohn2000}, we can think of $\intdiffopbas$ in two ways: Either we impose the linear relations on $\intdiffop$ or we merge them in with the Green's ideal---let us call these the a-posteriori and the combined approach, respectively.

For computational purposes, we need a \emph{ground simplifier} on the free algebra~\cite[p.~525]{Rosenkranz2005}, which we define here as a $K$-linear canonical simplifier for $\free/\basis$. Since all reduction rules of Table~\ref{fig:red-rules} are (bi)linear in $f,g \in \galg$, any ground simplifier descends to a canonical simplifier $\sigma$ on $\intdiffopbas$. In our implementations, $\sigma$ always operates by basis expansion (see below), but other strategies are conceivable. We can apply $\sigma$ either a-posteriori or combined:
 \begin{itemize}
 \item In the first case we apply $\sigma$ as a \emph{postprocessing step} after computing the normal forms with respect to Table~\ref{fig:red-rules}. We have chosen this approach in the upcoming \maple\ implementation~\cite{KorporalRegensburgerRosenkranz2010}.
\item In the \emph{combined approach}, $\sigma$ may be used at any point during a reduction along the rules of Table~\ref{fig:red-rules}. It may be more efficient, however, to use $\sigma$ on the rules themselves to create a new reduction system on $\free$; see below for an example. We have taken this approach in our earlier implementation~\cite{Rosenkranz2003,Rosenkranz2005} and in the current implementation.
\end{itemize}
Generally the first approach seems to be superior, at least when $\sigma$ tends to create large expressions that are not needed for the rewriting steps of Table~\ref{fig:red-rules}; this is what typically happens if the ground simplifier operates by basis expansion.

Assume now we choose a $K$-basis $(b_i)_{i \in I}$ of $\galg$. If $(\smash{\hat{b}}_i)_{i \in I}$ is the dual basis, we can describe the linear ideal more economically as
\begin{equation*}
\basis = \Big( f - \sum_{i \in I} \, \smash{\hat{b}}_i(f) \, b_i \mid f \in \galg \Big),
\end{equation*}
so the linear basis $(b_i)_{i \in I}$ gives rise to an ideal basis for $\basis$. Its generators $f - \sum_i \cdots$ can be oriented to create a ground simplifier $\sigma\colon f \mapsto \sum_i \cdots$ effecting \emph{basis expansion}.

If one applies now such a ground simplifier coming from a basis $(b_i)_{i \in I}$ in the \emph{combined approach}, one can restrict the generators of $\free$ to basis elements $b_i \in \galg$ rather than all $f \in \galg$, and the reduction rules can be adapted accordingly. For example, when $\galg$ contains the exponential polynomials, the Leibniz rule $\der f \rightarrow f \der + (\der \act f)$ gets instantiated for $f = x e^x$ as $\der (x e^x) \rightarrow (x e^x) \der + e^x + x$, where the right-hand side now has three instead of two monomials! This is why the choice of basis was built into the definition of the precursor of $\intdiffopbas$ as in~\cite{Rosenkranz2005}.

Before leaving this section on integro-differential operators, let us
mention some interesting \emph{related work} on these objects, carried
out from a more algebraic viewpoint. In his
papers~\cite{Bavula:2009b,Bavula2009a}, Bavula establishes an
impressive list of various (notably ring-theoretic) properties for
algebras of integro-differential operators. The setup considered in
these papers is, on the one hand, in many respects more general since it
deals with partial rather than ordinary differential operators but, on the other hand, the coefficients are restricted to polynomials.

Seen from the more applied viewpoint taken here, the most significant
difference is the lack of multiple point evaluations (and thus
boundary conditions). Apart from these obvious differences, there is
also a somewhat \emph{subtle difference} in the meaning of $\evl =
1-\cum\circ\der$ that we have tried to elucidate in a skew-polynomial
setting~\cite{RegensburgerRosenkranzMiddeke2009}. The upshot is that
while our approach views $\evl$ as a specific evaluation (the
prototypical example is given after Definition~\ref{def:ini-eval}), it
does not have a canonical action in V.~Bavula's setting (and neither
in our skew-polynomial approach). This is a subtle but far-reaching
difference that deserves future investigation.

\section{Applications for Boundary Problems}
\label{sec:bp}

In this section we combine the tools developed in the previous
sections to build an \emph{algorithm} for solving linear boundary problems
over an ordinary integro-differential algebra; see also~\cite{RosenkranzRegensburger2008a} for further details. We also outline a \emph{factorization method} along a given factorization of the defining differential operator applicable to boundary problems for both linear \emph{ordinary} and \emph{partial} differential equations; see~\cite{RegensburgerRosenkranz2009} in an abstract linear algebra setting and~\cite{RosenkranzRegensburgerTecBuchberger2009} for an overview.

For motivating our algebraic setting of boundary problems, let us consider
our standard example of an integro-differential algebra $(\galg,\der,\cum)$ with the integral operator
\[
\cum: f \mapsto \int_0^x f(\xi) \, d\xi
\]
for  $\galg= C^\infty[0,1]$. The \emph{simplest two-point boundary problem} reads then as follows:
Given $\f\in \galg$, find $u\in \galg$ such that
\begin{equation}
\label{eq:d2=0} \bvp{u''=f,}{u(0) = u(1) = 0.}
\end{equation}
In this and the subsequent examples, we let $D$ and $A$ denote respectively the derivation $\der$ and the integral operator $\cum$. Moreover, we denote by $L$ the corresponding evaluation $\evl$, which is the character
\[
L: f \mapsto f(0).
\]
To express boundary problems we need additionally the evaluation at the endpoint of the interval
\[
R: f \mapsto f(1).
\]
Note that $u$ is annihilated by any linear combination
of these functionals so that problem \eqref{eq:d2=0} can be
described by the pair $(D^2,[L,R])$, where $[L,R]$ is the subspace generated by
$L$, $R$ in the dual space $\Vd$ .

The solution algorithm presupposes a constructive fundamental system for the underlying
homogeneous equation but imposes no other conditions (in the
literature one often restricts to self-adjoint and/or second-order
boundary problems). This is always possible (relative to root computations) in the important special
case of LODEs with constant coefficient.

\subsection{The Solution Algorithm.}

In the following, we introduce the notion of \emph{boundary problem} in the context of
ordinary integro-differential algebras. Unless specified otherwise, all
integro-differential algebras in this section are assumed to be ordinary and over a fixed
field~$K$.

\begin{definition}
  \label{def:bvp}
  Let $(\galg, \der, \cum)$ be an ordinary integro-differential
  algebra. Then a \emph{boundary problem} of order $n$ is a pair $(T,\bspc)$, where $T \in \gdiffop$ is a regular differential operator  of order $n$ and $\bspc\subseteq \allscond$ is an $n$-dimensional subspace of boundary conditions.
\end{definition}

Thus a boundary problem is specified by a differential operator $T$
and a \emph{boundary space} $\bspc = [\beta_1, \ldots, \beta_n]$
generated by linearly independent boundary conditions $\beta_1, \ldots, \beta_n \in
\allscond$. In traditional notation, the boundary problem $(T, \bspc)$
is then given by
\begin{equation}
  \label{eq:bvp}
  \bvp{Tu=f,}{\beta_1 u = \cdots = \beta_n u = 0.}
\end{equation}
For a given boundary problem, we can restrict to a finite
subset $\Phi \subseteq \multfunc$, with the consequence that all
subsequent calculations can be carried out in $\intdiffopchar{\Phi}$
instead of $\intdiffop$. We will disregard this issue here for keeping
the notation simpler.

\begin{definition}
A boundary problem $(T, \bspc)$ is called \emph{regular} if for each $f \in \galg$ there exists a unique solution $u \in \galg$ in the sense of~\eqref{eq:bvp}.
\end{definition}

The condition requiring $T$ to have the same order as the dimension of
$\bspc$ in Definition~\ref{def:bvp} is only necessary but not
sufficient for ensuring regularity: the boundary conditions might
collapse on $\Ker{T}$. A simple example of such a \emph{singular
  boundary problem} is $(-D^2, [LD, RD])$ using the notation from
before; see also~\cite[p.~191]{Rosenkranz2005} for more details on
this particular boundary problem.

For an \emph{algorithmic} test of regularity, one may also apply the
usual regularity criterion for two-point boundary problems, as
described in~\cite{RegensburgerRosenkranz2009}. Taking any fundamental
system of solutions $u_1, \ldots, u_n$ for the homogeneous equation,
one can show that a boundary problem $(T, \bspc)$ is regular if and
only if the \emph{evaluation matrix}
\begin{displaymath}
  \beta(u) =
  \begin{pmatrix}
    \be_1(u_1) & \cdots & \be_1(u_n)\\
    \vdots & \ddots & \vdots\\
    \be_n(u_1) & \cdots & \be_n(u_n)
  \end{pmatrix} \in K^{n\times n}
\end{displaymath}
is regular.

For a regular boundary problem $(T, \bspc)$, we can now define its
\emph{Green's operator} $G$ as the linear operator mapping a given forcing
function~$f \in \galg$ to the unique solution $u \in \galg$
of~\eqref{eq:bvp}. It is characterized by the identities
\[
 TG = 1\quad\text{and}\quad \Img{G} = \orth{\bspc},
\]
where $\orth{\bspc}=\{u \in \galg \mid \be(u)=0\text{ for all } \be \in \B \}$ is the subspace of all ``functions'' satisfying the boundary conditions. We also write
\[
G = (T, \bspc)^{-1}
\]
for the Green's operator of $(T, \bspc)$.

The investigation of \emph{singular boundary problems}
(i.e. non-regular ones) is very enlightening but leads us too far
afield; we shall investigate it at another junction. Let us just
mention that it involves so-called modified Green's operators and
functions~\cite[p.~216]{Stakgold1979} and that is paves the way to an
interesting non-commutative analog of the classical Mikusi\'{n}ski
calculus~\cite{Mikusi'nski1959}.

We will now recast Theorem~\ref{thm:ivp} in the language of Green's
operators of initial value problems. Given a regular differential
operator $T$ of order $n$, the theorem implies that the initial value
problem $(T, [\evl, \evl \der, \ldots, \evl \der^{n-1}])$ is regular.
We call its Green's operator the \emph{fundamental right inverse} of
$T$ and denote it by $\fri{T}$.

\begin{corollary}
  \label{cor:fri}
  Let $(\galg, \der, \cum)$ be an ordinary integro-differential
  algebra and let $T \in \gdiffop$ be a regular differential operator of order $n$ with regular fundamental
  system $u_1, \ldots, u_n$. Then its fundamental right inverse is
  given by
  \begin{equation}
    \label{eq:fri}
    \fri{T} = \sum_{i=1}^n u_i \cum d^{-1} d_i \:\in\: \intdiffop,
  \end{equation}
  where $d, d_1, \ldots, d_n$ are as in Theorem~\ref{thm:ivp}.
\end{corollary}

Before turning to the solution algorithm for boundary problems, let us
also mention the following practical formula for specializing
Corollary~\ref{cor:fri} to the important special case of LODEs with
\emph{constant coefficients}, which could also be proved directly e.g.
via the Lagrange interpolation formula. For simplicity, we restrict
ourselves to the case where the characteristic polynomial is
separable.

\begin{corollary}
  \label{cor:fri-constant-coeff}
  Let $(\galg, \der, \cum)$ be an ordinary integro-differential
  algebra and consider the differential operator $T = (\der -
  \lambda_1) \cdots (\der - \lambda_n) \in \gdiffop$ with $\lambda_1,
  \ldots, \lambda_n \in K$ mutually distinct. Assume each $u' =
  \lambda_i \, u, \evl \cdot u = 1$ has a solution $u = e^{\lambda_i
    x} \in \galg$ with reciprocal $u^{-1} = e^{-\lambda_i x} \in
  \galg$. Then we have
  \begin{equation*}
    \fri{T} = \sum_{i=1}^n \mu_i \, e^{\lambda_i x} \cum e^{-\lambda_i
      x},
  \end{equation*}
  where $\mu_i^{-1} = (\lambda_i - \lambda_1) \cdots (\lambda_i -
  \lambda_{i-1}) (\lambda_i - \lambda_{i+1}) \cdots (\lambda_i -
  \lambda_n)$.
\end{corollary}
\begin{rrproof}
  Let us write $V$ for the $n \times n$ Vandermonde determinant in
  $\lambda_1, \ldots, \lambda_n$ and $V_i$ for the $(n-1) \times
  (n-1)$ Vandermonde determinant in $\lambda_1, \ldots, \lambda_{i-1},
  \lambda_{i+1}, \ldots, \lambda_n$. Evaluating the quantities
  of~\eqref{eq:fri}, one sees immediately that
  \begin{equation*}
    d = e^{(\lambda_1 + \cdots \lambda_n)x} \, V
    \quad\text{and}\quad
    d_i = (-1)^{n+i} e^{(\lambda_1 + \cdots
      + \lambda_{i-1} + \lambda_{i+1} + \cdots + \lambda_n) x} \, V_i.
  \end{equation*}
  Hence we have $d_i/d = (-1)^{n+i} \, e^{-\lambda_i x} \: V_i/V$.
  Using the well-known formula for the Vandermonde determinant, one
  obtains $d_i/d = \mu_i \, e^{-\lambda_i x}$, and now the result
  follows from Corollary~\ref{cor:fri}.
\end{rrproof}

Summarizing our earlier results, we can now give a \emph{solution
  algorithm} for computing $G = (T, \bspc)^{-1}$, provided we have a
regular fundamental system $u_1, \ldots, u_n$ for $Tu=0$ and a
$K$-basis $\beta_1, \ldots, \beta_n$ for~$\bspc$. The algorithm
proceeds in three steps:
\begin{enumerate}
\setlength{\itemsep}{3pt}
\item Construct the fundamental right inverse $\smash{\fri{T} \in
    \intdiffop}$ as in Corollary~\ref{cor:fri}.
\item Determine the projector $P = \smash{\sum_{i=1}^n} u_i
  \tilde\beta_i \in \intdiffop$ as in~\eqref{eq:projector}.
\item Compute $G = (1-P) \fri{T} \in \intdiffop$.
\end{enumerate}

\begin{theorem}
  \label{thm:green-op}
  The above algorithm computes the Green's operator $G \in \intdiffop$
  for any regular boundary problem $(T, \bspc)$.
\end{theorem}
\begin{rrproof}
 See~\cite{RosenkranzRegensburger2008a}.
\end{rrproof}

The computation of \emph{Green's operators} for boundary problems for ODEs using the above algorithm takes on the following concrete form in \tma\ code.
\begin{mmacode}
\includegraphics[scale=0.75]{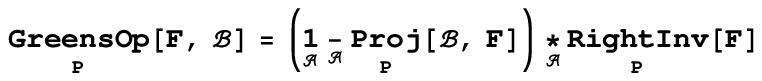}
\end{mmacode}
\noindent Here $\mathcal{B}$ is a basis for the boundary space and
$\mathtt{F}$ a regular fundamental system.

Let us consider again example~\eqref{eq:d2=0}: Given $\f\in \galg$, find $u\in \galg$ such that
\begin{equation*}
\bvp{u''=f,}{u(0) = u(1) = 0.}
\end{equation*}
The Green's operator $G$ of the boundary problem can be obtained by our implementation via the following computation
\begin{mmacode}
\includegraphics[scale=0.75]{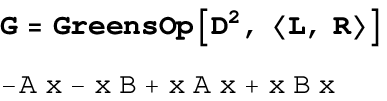}
\end{mmacode}
where we use the notation from before: $Au = \int_0^{x} u(\xi) \, d\xi$, $Lu = u(0)$, $Ru = u(1)$ and in addition, $Bu = \int_x^{1} u(\xi) \, d\xi$. The corresponding Green's function is computed in an immediate postprocessing step:
\begin{mmacode}
\includegraphics[scale=0.75]{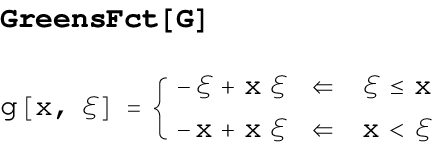}
\end{mmacode}
As noted in~\cite{Rosenkranz2005}, the Green's function provides a canonical form for the Green's operator. Moreover, one can obtain the function $u(x)$ and thus solve the boundary problem through knowledge of the Green's function in the following identity:
\[
u(x) = G f(x) =\int_0^1g(x,\xi)f(\xi)\,d\xi.
\]
By replacing the Green's function obtained above in the latter integral we obtain
\[
u(x) = (x-1) \int_0^x \xi f(\xi) \, d \xi + x \int_x^1 (\xi -1) f(\xi) \, d\xi.
\]

Furthermore, we can look at some specific instances of the forcing function $f(x)$. Let us first consider the simple example $f(x) = x$. By an immediate calculation, we obtain the expression for the action of $G$ on $f(x)$, which is $u(x)$:
\begin{mmacode}
\includegraphics[scale=0.75]{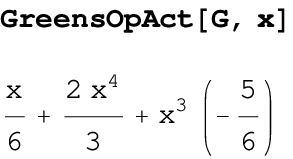}
\end{mmacode}
The expression for the solution function $u(x)$ can easily become more complicated, as we can see in the next example, where we consider the instance
\[
f(x) = e^{2 x} + 3 x^2 {\sin x}^3.
\]

\noindent Relying on \mma\ for handling symbolic integration, we obtain:
\begin{mmacode}
\includegraphics[scale=0.66]{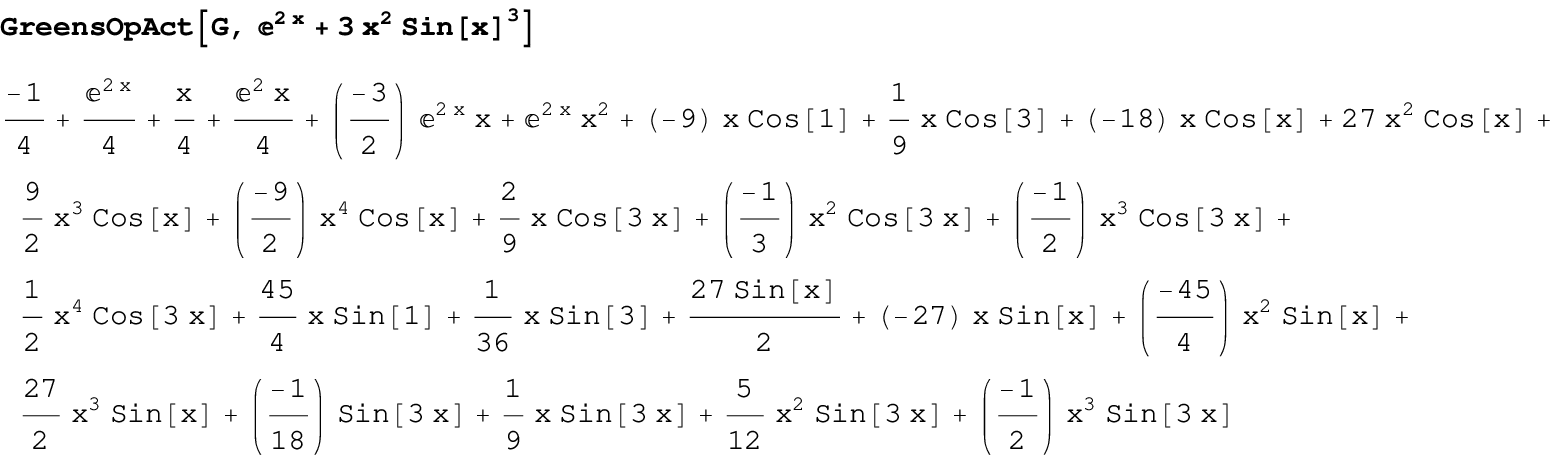}
\end{mmacode}

\noindent As a last example, let us consider $f(x) = \sin \sin x$. As we can notice below, it cannot be integrated with \mma:
\begin{mmacode}
\includegraphics[scale=0.75]{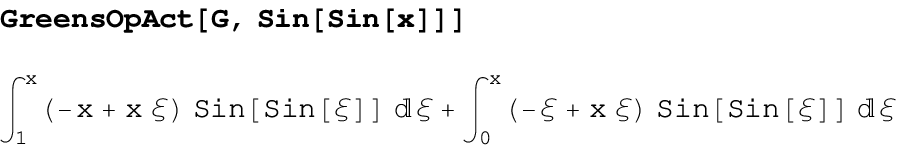}
\end{mmacode}

\noindent In order to carry out the integrals involved in the application of the Green's operator to a forcing function, one can use any numerical quadrature method (as also available in many computer algebra systems).

\subsection{Composing and Factoring Boundary Problems.}

In the following, we discuss the \emph{composition} of boundary problems corresponding to their Green's operators. We also describe how \emph{factorizations} of a boundary problem along a given factorization of the defining operator can be characterized and constructed. We refer again to~\cite{RosenkranzRegensburger2008a,RegensburgerRosenkranz2009} for further details. We assume that all operators are defined on suitable spaces such that the composition is well-defined. It is worth mentioning that the following approach works in an \emph{abstract setting}, which includes in particular boundary problems for linear partial differential equations (LPDEs) and systems thereof; for simplicity, we will restrict ourselves in the examples to the LODE setting.

\begin{definition}
\label{def:compo}
We define the \emph{composition} of boundary problems $(T_1, \B_1)$ and $(T_2, \B_2)$ by
\[
(T_1, \B_1) \circ (T_2, \B_2)= (T_1 T_2, \B_1 \cdot T_2 + \B_2).
\]
\end{definition}

So the boundary conditions from the first boundary problem are ``translated'' by the operator from the second problem. The composition of boundary problems is associative but in general not commutative. The next proposition tells us that the composition of boundary problems preserves \emph{regularity}.

\begin{proposition}
\label{prop:compo}
Let $(T_1,\B_1)$ and $(T_2,\B_2)$ be regular boundary
problems with Green's operators $G_1$ and $G_2$. Then $(T_1,\B_1) \circ (T_2,\B_2)$ is regular with Green's operator $G_2 G_1$ so that
\[
((T_1,\B_1) \circ (T_2,\B_2))^{-1}=(T_2,\B_2)^{-1} \circ (T_1,\B_1)^{-1}.
\]
\end{proposition}

The simplest example of composing two boundary (more specifically, initial value) problems for ODEs is the following. Using the notation from before, one sees that
\[
(D,[L])\circ(D,[L]) = (D^2, [LD] + [L])=(D^2,
[L,LD]).
\]

Let now $(T,\B)$ be a boundary problem and assume that we have a factorization $T=T_1T_2$ of the defining differential operator. We refer to~\cite{RosenkranzRegensburger2008a,RegensburgerRosenkranz2009} for a characterization and construction of all factorizations
\[
(T,\B)=(T_1, \B_1) \circ (T_2, \B_2)
\]
into boundary problems. In particular, if $(T,\B)$ is regular, it can be factored into regular boundary problems: the left factor $(T_1, \B_1)$ is \emph{unique}, while for the right factor $(T_2,\B_2)$ we can choose any subspace $\B_2 \le \B$ that makes the problem regular. We can compute the uniquely determined boundary conditions for the left factor by $\B_1=\B \cdot G_2$, where $G_2$  is the Green's operator for some regular right factor $(T_2,\B_2)$.
By factorization, one can split a problem of higher order into subproblems of \emph{lower} order, given a factorization of the defining operator. For algorithms and results about factoring
ordinary differential operators we refer to~\cite{PutSinger2003,Schwarz1989,Tsarev1996}.

Given a fundamental system of the differential operator $T$ and a right inverse of $T_2$, one can factor boundary problems in an algorithmic way as shown in~\cite{RegensburgerRosenkranz2009} and in an integro-differential algebra~\cite{RosenkranzRegensburger2008a}. As described in~\cite{RosenkranzRegensburgerTecBuchberger2009}, we can also compute boundary conditions $\B_2\le\B$ such that $(T_2,\B_2)$ is a regular right factor, given only a fundamental system of $T_2$. The unique left factor can be then computed as explained above.
This allows us to factor a boundary problem if we can factor the differential operator and compute a fundamental system of only one factor. The remaining lower order problems can then also be solved by numerical methods.

Here is how we can compute the boundary conditions of the left and right factor problems for the boundary problem $(D^2,[L,R])$ from previous example~\eqref{eq:d2=0}, along the trivial factorization with $T_1=T_2=D$. The indefinite integral $A=\int_0^x$ is the Green's operator for the regular right factor $(D,[L])$.
\begin{mmacode}
\includegraphics[scale=0.75]{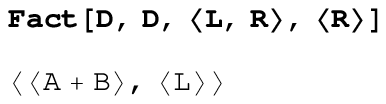}
\end{mmacode}
This factorization reads in traditional notation as
\begin{displaymath}
   \bvp{u'=f}{\int_0^1 u(\xi) \, d\xi = 0}\circ
   \bvp{u'=f}{u(0) = 0} =
   \bvp{u''=f}{u(0) = u(1) = 0}.
 \end{displaymath}
Note that the boundary condition for the unique left factor is an integral (Stieltjes) boundary condition.

We consider as a second example the fourth order boundary problem~\cite[Ex.~33]{RosenkranzRegensburger2008a}:
\begin{equation}
\label{eq:d4=0} \bvp{u'''' + 4 u=f,}{u(0) = u(1) = u'(0) = u'(1) = 0.}
\end{equation}
%with the factorization of boundary conditions:
Factoring the boundary problem along $D^4+4 = (D^2-2i)(D^2+2i)$, we obtain the following
boundary conditions for the factor problems.
\begin{mmacode}
\includegraphics[scale=0.75]{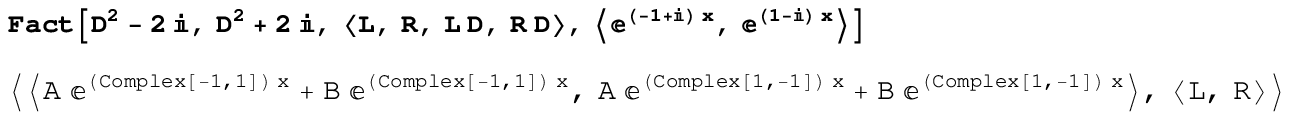}
\end{mmacode}

\section{Integro-Differential Polynomials}
\label{sec:intdiffpol}

In this section, we describe the algebra of \emph{integro-differential polynomials}~\cite{RosenkranzRegensburger2008} obtained by adjoining an indeterminate function to a given integro-differential algebra $(\galg,\der,\cum)$. Intuitively, these are all terms  that one can build using an indeterminate $u$, coefficient functions $f \in \galg$ and the operations $+,\cdot,\der,\cum$, identifying two terms if one can be derived from the other by the axioms of integro-differential algebras and the operations in $\galg$. A typical term for $(K[x],\der,\cum)$ looks like this:
\[
(4 u u' \cum (x+3) u'^3) (u' \cum u''^2) + \cum x^6 u u''^5 \cum (x^2 + 5x) u^3 u'^2\cum u
\]
From the computational point of view, a fundamental problem is to find a canonical simplifier (see Sect.~\ref{sec:theorema}) on these objects. For example, the above term can be transformed to
\begin{multline*}
4 u u'^2 \cum x u'^3 \cum u''^2 + 4 u u'^2 \cum u''^2 \cum x u'^3 + 12 u u'^2 \cum u'^3 \cum u''^2 + 12 u u'^2 \cum u''^2 \cum u'^3 \\
+ \cum x^6 u u''^5 \cum x^2 u^3 u'^2\cum u + 5 \cum x^6 u u''^5 \cum x u^3 u'^2\cum u.
\end{multline*}
by the Baxter axiom and the $K$-linearity of the integral.

As outlined in the next subsection, a notion of polynomial can be constructed for any variety in the sense of \emph{universal algebra}. (In this general sense, an algebra is a set with an arbitrary number of operations, and a variety is a collection of such algebras satisfying a fixed set of identities. Typical examples are the varieties of groups, rings, and lattices.)

For sample computations in the algebra of integro-differential polynomials, we use a \emph{prototype implementation} of integro-differential polynomials, based on the \tma\ functor mechanism (see Sect.~\ref{sec:theorema}).

\subsection{Polynomials in Universal Algebra.}

In this subsection, we describe the idea of the general construction of polynomials in universal algebra~\cite{Hule1969}.
We refer to~\cite{LauschNoebauer1973} for a comprehensive treatment; see also the surveys~\cite{AichingerPilz2003,BuchbergerLoos1983}. For the basic notions in universal algebra used below, see for example~\cite{BaaderNipkow1998} or~\cite[Ch.~1]{Cohn2003a}.

Let $\mathcal{V}$ be a variety defined by a set $\identities$ of identities over a signature $\Sigma$. Let $A$ be a fixed ``coefficient domain'' from the variety $\mathcal{V}$, and let $X$ be a set of indeterminates (also called ``variables''). Then all terms in the signature $\Sigma$ with constants (henceforth referred to as coefficients) in $A$ and indeterminates in $X$ represent the same polynomial if their equality can be derived in finitely many steps from the identities in $\identities$ and the operations in~$A$. The set of all such terms $\term{\Sigma}{A}{X}$ modulo this congruence $\equiv$ is an algebra in $\mathcal{V}$, called the \emph{polynomial algebra} for $\mathcal{V}$ in $X$ \mbox{over $A$} and denoted by $\polalg$.

The polynomial algebra $\polalg$ contains $A$ as a subalgebra, and $A \cup X$ is a generating set. As in the case of polynomials for commutative rings, we have the \emph{substitution homomorphism} in general polynomial algebras. Let $B$ be an algebra in $\mathcal{V}$. Then given a homomorphism $\varphi_1\colon A \rightarrow B$ and a map $\varphi_2\colon X \rightarrow B$, there exists a unique homomorphism
 \[
 \varphi \colon \polalg \rightarrow B
 \]
such that $\varphi(a)=\varphi_1(a)$ for all $a \in A$ and $\varphi(x)=\varphi_2(x)$ for all $x\in X$. So in a categorical sense the polynomial algebra $\polalg$ is a free product of the coefficient algebra $A$ and the free algebra over $X$ in $\mathcal{V}$; see also~\cite{AichingerPilz2003}.

For computing with polynomials, we will construct a \emph{canonical simplifier} on $\polalg$ with associated system of canonical forms $\canf$. As explained before (Sect.~\ref{sec:theorema}), the canonical simplifier provides for every polynomial in $\polalg$, represented by some term $T$, a canonical form $R\in \canf$ that represents the same polynomial, with different terms in $\canf$ representing different polynomials; see also~\cite[p.~23]{LauschNoebauer1973}.

The set $\canf$ must be large enough to generate all of $\polalg$ but small enough to ensure \emph{unique representatives}. The latter requirement can be ensured by endowing a given set of terms with the structure of an algebra in the underlying variety.

\begin{proposition}
\label{prop:canforms}
Let $\mathcal{V}$ be a variety over a signature $\Sigma$, let $A$ be an algebra in $\mathcal{V}$ and $X$ a set of indeterminates.  If $\canf\subseteq \term{\Sigma}{A}{X}$ is a set of terms with $A \cup X \subseteq \canf$ that can be endowed with the structure of an algebra in $\mathcal{V}$, then different terms in $\canf$ represent different polynomials in $\polalg$.
\end{proposition}

\begin{proof}
Since $\canf$ can be endowed with the structure of an algebra in the variety $\mathcal{V}$ and $A \cup X \subseteq \canf$, there exists a unique substitution homomorphism
  \begin{displaymath}
    \phi\colon \polalg \rightarrow \canf
  \end{displaymath}
  such that $\phi(a) = a$ for all $a \in A$ and $\phi(x) = x$ for all $x \in X$. Let
  \begin{displaymath}
    \pi\colon \canf \rightarrow \polalg
  \end{displaymath}
 denote the restriction of the canonical map associated with $\equiv$. Then we have $\phi \circ \pi (R) = R$ for all $R \in \canf$, so $\pi$ is injective, and different terms in $\canf$ indeed represent different polynomials.
 \qed
 \end{proof}

As a well-known example, take the \emph{polynomial ring} $R[x]$ in one indeterminate $x$ over a commutative ring $R$, which is $\polalg$ for $A=R$ and $X=\{x\}$ with $\mathcal{V}$ being the variety of commutative unital rings. The set of all terms of the form $a_nx^n+\cdots+a_0$ with coefficients $a_i \in R$ and $a_n\neq 0$ together with $0$ is a system of canonical forms for $R[x]$. One usually defines the polynomial ring directly in terms of these canonical forms. Polynomials for groups, bounded lattices  and Boolean algebras are discussed in~\cite{LauschNoebauer1973} along with systems of canonical forms.

\subsection{Differential Polynomials.}
For illustrating the general construction described above, consider the algebra of \emph{differential polynomials} over a commutative differential $K$-algebra $(\galg,\der)$ in one indeterminate function $u$, usually denoted by $\intdiffpol{u}$. Clearly this is $\polalg$ for $A = \galg$ and $X = \{u\}$ with $\mathcal{V}$ being the variety of differential $K$-algebras. Terms are thus built up with the indeterminate $u$, coefficients from $\mathcal{F}$ and the operations $+,\cdot,\der$; a typical example being
\[
\der^2 (f_1 u^2 + u) \der (f_2 u^3)+\der^3(f_3 u).
\]
By applying the Leibniz rule and the linearity of the derivation, it is clear that every polynomial is congruent to a $K$-linear combination of terms of the form
\begin{equation}
\label{eq:diffpol}
f \prod_{i=0}^{\infty} u_i^{\beta_i},
\end{equation}
where $f \in \galg$, the notation $u_n$ is short for $\der^n(u)$, and only finitely many $\beta_i \in \N$ are nonzero. In the following, we use the multi-index notation $f u^{\beta}$ for terms of this form. For instance, $u^{(1,0,3,2)}$ is the multi-index notation for $u (u'')^3 (u''')^2$. The order of a differential monomial $u^{\beta}$ is given by the highest derivative appearing in $u^{\beta}$ or $-\infty$ if $\beta=0$.

Writing $\canf$ for the set of all $K$-linear combinations of terms of the form~\eqref{eq:diffpol}, we already know that every polynomial is congruent to a term in $\canf$. When $\galg = K[x]$, a typical element of $\canf$ is given by
\begin{equation*}
(3 x^3 + 5 x) \, u^{(1,0,3,2)} + 7 x^5 u^{(2,0,1)} + 2 x u^{(1,1)}.
\end{equation*}
To show that $\canf$ is a system of canonical forms for $\intdiffpol{u}$, by Proposition~\ref{prop:canforms} it suffices to endow $\canf$ with the structure of a commutative differential algebra. As a commutative algebra, $\canf$ is just the polynomial algebra in the infinite set of indeterminates $u_0, u_1, u_2, \ldots$. For defining a derivation in a commutative algebra, by the Leibniz rule and $K$-linearity, it suffices to specify it on the generators. Thus $\canf$ becomes a differential algebra by setting $\der(u_k)=u_{k+1}$. One usually defines the differential polynomials directly in terms of these canonical forms, see for example~\cite{Kolchin1973}.

\subsection{Integro Polynomials.}
We outline the \emph{integro polynomials} over a Rota-Baxter algebra as in Definition~\ref{def:Rota-Baxter}. This is related to the construction of free objects in general Rota-Baxter algebras; we refer to~\cite{GuoSit2010} for details and references. By iterating the Baxter axiom~\eqref{eq:baxter-axiom}, one obtains a generalization that is called the \emph{shuffle identity} on $\galg$:
\begin{equation}
\label{eq:shuffle}
(\cum f_1 \cum \cdots \cum f_m) \cdot
(\cum g_1 \cum \cdots \cum g_n)
= \sum \cum h_1 \cum \cdots \cum h_{m+n}
\end{equation}
\noindent Here the sum ranges over all shuffles of $(f_1,\ldots,f_m)$ and $(g_1,\ldots,g_n)$; see~\cite{Ree1958,Reutenauer1993,Rota1998} for details. The sum consists of $\tbinom{m + n}{n}$ shuffles, obtained by ``shuffling'' together $\cum f_1 \cum \cdots \cum f_m$ and $\cum g_1 \cum \cdots \cum g_n$ as words over the letters $\cum f_i$ and $\cum g_j$, such that the inner order in the words is preserved. For instance, we have
\[
(\cum f_1 \cum f_2) \cdot (\cum g_1) = \cum f_1 \cum f_2 \cum g_1 + \cum f_1 \cum g_1 \cum f_2 + \cum g_1 \cum f_1 \cum f_2.
\]
for the simple $m=2, n=1$ case.

The integro polynomials over $\galg$ are defined as $\polalg$ for $A=\galg$ and $X = \{u\}$ with $\mathcal{V}$ being the variety of Rota-Baxter algebras over $K$. The full construction of the canonical forms for integro polynomials is included in the following subsection. But it is clear that by \emph{expanding products of integrals} by the shuffle identity, every integro polynomial is congruent to a $K$-linear combination of terms of the form
\begin{equation}
\label{eq:canfintpol}
f u^k \cum f_1 u^{k_1} \cum \cdots \cum f_m u^{k_m}
\end{equation}
with $f, f_1, \ldots, f_m \in \galg$ and $k, k_1, \ldots, k_m \in \N$. However, they cannot be canonical forms, since terms like $\cum (f+g) u$ and $\cum fu + \cum gu$ or $\cum \lambda f u$ and $\lambda \cum f u$ represent the same polynomials.

Writing $\canf$ for the set of all $K$-linear combinations of terms of the form~\eqref{eq:canfintpol}, the multiplication of two elements of $\canf$ can now be defined via~\eqref{eq:shuffle} as follows. Since the product of~\eqref{eq:canfintpol} with another term $g u^l \cum g_1 u^{l_1} \cum \cdots \cum g_n u^{l_n}$ should clearly be given by $fg \, u^{k+l} (\cum f_1 u^{k_1} \cum \cdots \cum f_m u^{k_m}) \, (\cum g u^l g_1 u^{l_1} \cum \cdots f_n u^{l_n})$, it remains to define the so-called \emph{shuffle product} on integral terms (those having the form~\eqref{eq:canfintpol} with $f=1$ and $k=0$). This can be achieved immediately by using~\eqref{eq:shuffle} with $f_i u^{k_i}$ and $g_j u^{l_j}$ in place of $f_i$ and $g_j$, respectively. It is easy to see that the shuffle product is commutative and distributive with respect to addition.

The shuffle product can also be defined \emph{recursively}~\cite{Reutenauer1993}. Let $J$ and $\tilde J$ range over integral terms (note that $1$ is included as the special case of zero nested integrals). Then we have
\begin{equation}
\label{eq:shuffle2}
(\cum f u^k J) \cdot (\cum \tilde{f} u^{\tilde{k}} \tilde{J}) =  (\cum f u^k) \nest J \cdot (\cum \tilde{f} u^{\tilde{k}} \tilde{J}) + (\cum \tilde{f} u^{\tilde{k}}) \nest (\cum f u^k J) \cdot \tilde{J},
\end{equation}
where $\nest\colon \canf \times \canf \rightarrow \canf$ denotes the operation of nesting integrals (with the understanding that $\cdot$ binds stronger than $\nest$), defined on basis vectors by
\begin{equation}
\label{eq:nest}
\cum F_1 \cum \cdots \cum F_m \nest \cum G_1 \cum \cdots \cum G_n
= \cum F_1 \cum \cdots \cum F_m \cum G_1 \cum \cdots \cum G_n,
\end{equation}
and extended bilinearly to all of $\canf$. Here $F_i$ and $G_j$ stand for $f_i u^{k_i}$ and $g_j u^{l_j}$, respectively. For example, $\cum F_1 \cum F_2$ and $\cum G_1$ can be multiplied as
\begin{multline*}
(\cum F_1) \nest (\cum F_2) \cdot (\cum G_1)+ (\cum G_1) \nest 1 \cdot (\cum F_1 \cum F_2)
= (\cum F_1) \nest (\cum F_2 \cum G_1 + \cum G_1 \cum F_2) \\
+ (\cum G_1) \nest (\cum F_1 \cum F_2)= \cum F_1 \cum F_2 \cum G_1 + \cum F_1 \cum G_1 \cum F_2 + \cum G_1 \cum F_1 \cum F_2,
\end{multline*}
analogous to the previous computation.

\subsection{Representing Integro-Differential Polynomials.}

In the following, we describe in detail the universal algebra construction of the integro-differential polynomials and their canonical forms. We refer to~\cite{GuoKeigher2008,GuoSit2010a} for the related problem of constructing free objects in differential Rota-Baxter algebras. We consider the variety of integro-differential algebras. Its \emph{signature} $\Sigma$ contains: the ring operations, the derivation~$\der$, the integral~$\cum$, the family of unary ``scalar multiplications'' $(\cdot_\lambda)_{\lambda \in K}$, and for convenience we also include the evaluation $\evl$. The \emph{identities} $\identities$ are those of a \mbox{$K$-algebra}, then \mbox{$K$-linearity} of the three operators $\der$, $\cum$, $\evl$, the Leibniz rule~\eqref{eq:leibniz}, the section axiom \eqref{eq:section-axiom}, the Definition~\ref{def:ini-eval} of the
evaluation, and the differential Baxter axiom \eqref{eq:int-by-parts}.

\begin{definition}
 Let $(\galg,\der,\cum)$ be an integro-differential algebra. Then the algebra of \emph{integro-differential polynomials} in $u$ over $\galg$, denoted by $\intdiffpol{u}$ in analogy to the differential polynomials, is the polynomial algebra $\polalg$ for $A=\galg$ and $X = \{u\}$ with $\mathcal{V}$ being the variety of integro-differential algebras over $K$.
\end{definition}

Some identities following from~$\identities$ describe \emph{basic interactions} between operations in $\galg$: the pure Baxter axiom \eqref{eq:baxter-axiom}, multiplicativity of the evaluation~\eqref{eq:mult}, the identities
\begin{equation}
 \label{eq:ad_identities}
 \evl^2=\evl, \; \; \; \der \evl = 0, \; \; \; \evl\cum = 0, \; \; \; \cum (\evl f) g = (\evl f)
 \cum g,
\end{equation}
and the variant \eqref{eq:int-by-parts-eval} of the differential Baxter axiom connecting all three operations.

We need to introduce some \emph{notational conventions}. We use $f,g$ for coefficients in $\galg$, and $V$ for terms in $\term{\Sigma}{\galg}{\{u\}}$. As for differential polynomials, we write $u_n$ for the $n$th derivative of $u$. Moreover, we write
\[
V(0) \;\text{for}\; \evl(V) \quad\text{and}\quad u(0)^{\alpha} \;\text{for}\; \prod_{i=0}^{\infty} u_i(0)^{\alpha_i},
\]
where $\alpha$ is a multi-index.

As a first step towards canonical forms, we describe below a system of terms that is \emph{sufficient for representing} every integro-differential polynomial (albeit not uniquely as we shall see presently).

\begin{lemma}
  \label{lem:sumf}
  Every polynomial in $\intdiffpol{u}$ can be represented by a finite
  $K$-linear combination of terms of the form
  \begin{equation}
    \label{eq:sumf}
    f u(0)^{\alpha}u^{\beta}\cum f_{1} u^{\gamma_1} \cum \cdots \cum f_{n}
    u^{\gamma_n},
  \end{equation}
  where $f, f_1,\ldots,f_n \in \galg$, and each multi-index as well as $n$ may be zero.
\end{lemma}

\begin{proof}
%%to be done%%
The proof is done by induction on the structure of terms, using the above identities \eqref{eq:baxter-axiom},~\eqref{eq:mult},~\eqref{eq:shuffle} and~\eqref{eq:ad_identities} of integro-differential algebras.
\qed
\end{proof}

\noindent With the aid of the previous lemma we can determine the constants of $\intdiffpol{u}$.

\begin{proposition}
  \label{prop:const}
  Every constant in $\intdiffpol{u}$ is represented as a finite sum $\sum_{\alpha} c_{\alpha} u(0)^{\alpha}$ with constants $c_{\alpha}$ in $\galg$.
\end{proposition}

\begin{proof}
By the identity $\cum \der = 1 - \evl$, a term $V$ represents a constant in $\intdiffpol{u}$ if and only if $\evl(V)\equiv V$. Since $V$ is congruent to a finite sum of terms of the form~\eqref{eq:sumf} and since $\Img{\evl}=\const$, the identities for $\evl$ imply that $V$ is congruent to a finite sum of terms of the form $c_{\alpha} u(0)^{\alpha}$.
\qed
\end{proof}

The above representation~(\ref{eq:sumf}) of the integro-differential polynomials is \emph{not unique} since for example when trying to integrate differential polynomials by using integration by parts, terms like
\[
\cum f u' \quad\text{and}\quad f u-\cum f' u - f(0) \, u(0)
\]

\noindent are equivalent. It becomes even more tedious to decide that, for instance,
\[
2x \, u(0)^{(3,1)}u^{(1,3,0,4)}\cum (2x^3+3x) \, u^{(1,2,3)}\cum (x+2) \, u^{(2)}
\]
and
\begin{eqnarray*}
& 4x \, u(0)^{(3,1)}u^{(1,3,0,4)}\cum x^3u^{(1,2,3)}\cum x \, u^{(2)} + 6x \, u(0)^{(3,1)}u^{(1,3,0,4)}\cum x \, u^{(1,2,3)}\cum (x+2) \, u^{(2)} \\
& + 12x \, u(0)^{(3,1)}u^{(1,3,0,4)}\cum x \, u^{(1,2,3)}\cum u^{(2)}
\end{eqnarray*}
represent the same polynomial. In general, the following identity holds:

\begin{lemma}
  \label{lem:diff_baxter}
  We have
    \begin{equation}
    \label{eq:diff_baxter_terms}
    \cum V u_k^{\beta_k} u_{k+1}
    \equiv \frac{1}{\beta_k+1}\left(Vu_{k}^{\beta_k + 1}-\cum V'
      u_k^{\beta_k + 1} - V(0) \, u_k(0)^{\beta_k + 1}\right)
  \end{equation}
  where $k,\beta_k \geq 0$.
\end{lemma}

\begin{proof}
  Using \eqref{eq:int-by-parts-eval} and the Leibniz rule, the left-hand side becomes
  \begin{equation*}
    \cum (V u_k^{\beta_k}) (u_{k})' \equiv Vu_{k}^{\beta_k + 1}-\cum V' u_k^{\beta_k + 1} -\beta_k \cum V
    u_k^{\beta_k}u_{k+1} -  V(0) \, u_k(0)^{\beta_k + 1},
  \end{equation*}
  and the equation follows by collecting the $\cum V u_k^{\beta_k}u_{k+1}$ terms.
\qed
\end{proof}

The important point to note here is that if the highest derivative in the differential monomial $u^\beta$ of order $k+1$ appears \emph{linearly}, then the term $\cum f u^\beta$ is congruent to a sum of terms involving differential monomials of order at most $k$. This observation leads us to the following classification of monomials; confer also~\cite{Bilge1992,GelfandDikiui1976}.
\begin{definition}
\label{def:class}
A differential monomial $u^\beta$ is called \emph{quasiconstant} if $\beta = 0$, \emph{quasilinear} if $\beta \not = 0$ and the highest derivative appears linearly; otherwise it is called \emph{functional}. An integro-differential monomial~(\ref{eq:sumf}) is classified according to its outer differential monomial $u^\beta$, and its order is defined to be that of $u^\beta$.
\end{definition}

\begin{proposition}
  \label{prop:surj}
  Every polynomial in $\intdiffpol{u}$ can be represented by a $K$-linear combination of terms of the form
  \begin{equation}
    \label{eq:sumnf}
    f u(0)^{\alpha}u^{\beta}\cum f_1 u^{\gamma_1} \cum \cdots \cum f_n
    u^{\gamma_n},
  \end{equation}
where $f, f_1,\ldots,f_n \in \galg$, the multi-indices $\alpha,\beta$ as well as $n$ may be zero and the $u^{\gamma_1},\ldots,u^{\gamma_n}$ are functional.
\end{proposition}
\begin{proof}
By Lemma~\ref{lem:sumf} we can represent every polynomial in $\intdiffpol{u}$ as a $K$-linear combination of terms of the form
 \begin{equation}
    \label{eq:sumnof}
    f u(0)^{\alpha}u^{\beta}\cum f_{1} u^{\gamma_1} \cum \cdots \cum f_{n} u^{\gamma_n},
  \end{equation}
where the multi-indices and $n$ can also be zero. Let us first prove by induction on depth that every term can be written as in~\eqref{eq:sumnof} but with nonzero multi-indices $\gamma_k$. The base case $n=1$ is trivial since $\cum f_1$ can be pulled to the front. For the induction step we proceed from right to left, using the identity
\begin{equation*}
\cum f \cum V \equiv \cum f \cdot \cum V -\cum V \cum f
\end{equation*}
implied by the pure Baxter axiom~\eqref{eq:baxter-axiom}.

For proving that every multi-index $\gamma_k$ in~\eqref{eq:sumnof} can be made functional, we use noetherian induction with respect to the preorder on $J = \cum f_1 u^{\gamma_1} \cum \cdots \cum f_n u^{\gamma_n}$ that first compares depth and then the order of $ u^{\gamma_1}$. One readily checks that the left-hand side of~\eqref{eq:diff_baxter_terms} is greater than the right-hand side with respect to this preorder, provided that $V$ is of this form.

 Applying Lemma~\ref{lem:diff_baxter} inductively, a term $\cum f_1 u^{\gamma_1}$ is transformed to a sum of terms involving only integral terms with functional differential monomials, and the base case $n=1$ follows. As induction hypothesis, we assume that all terms that are smaller than $J = f u(0)^{\alpha}u^{\beta}\cum f_{1} u^{\gamma_1} \cum \cdots \cum f_{n} u^{\gamma_n}$ can be written as a sum of terms involving only functional monomials. Since $\cum f_2 u^{\gamma_2} \cum \cdots \cum f_n u^{\gamma_n}$ is smaller than $J$, it can be written as sum of terms involving only functional monomials; we may thus assume that $u^{\gamma_2}, \ldots, u^{\gamma_n}$ are all functional. Since $\gamma_1$ is nonzero, we are left with the case when $u^{\gamma_1}$ is quasilinear. Applying again Lemma~\ref{lem:diff_baxter} inductively, we can replace $u^{\gamma_1}$ in $J$ by a sum of terms involving only integral terms with functional differential monomials. The induction step follows then by the linearity of $\cum$.
\qed
\end{proof}

For \emph{implementing} the integro-differential polynomials in \tma\, we use the functor hierarchy described in Sect.~\ref{sec:theorema}. The multi-index representation $u^{\beta}$ for terms of the form~\eqref{eq:diffpol} is realized by the monoid $\N^*$ of natural tuples with finitely many nonzero entries, generated by a functor named $\mathtt{TuplesMonoid}$. The nested integrals $\cum f_1 u^{\gamma_1} \cum \cdots \cum f^n u^{\gamma_n}$ are represented as lists of pairs of the form $\langle f_k,\gamma_k \rangle$, with $f_k \in \galg$ and $\gamma_k \in \N^*$. The terms of the form~\eqref{eq:sumf} are then constructed via a cartesian product of monoids as follows:

\begin{mmacode}
\includegraphics[scale=0.75]{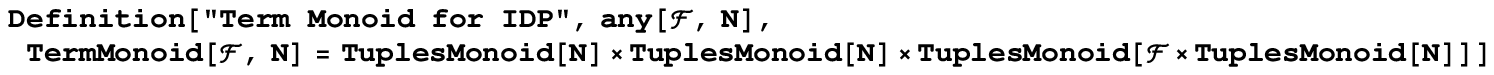}
\end{mmacode}

Using this construction, the integro-differential polynomials are built up by the functor $\mathtt{FreeModule[\galg,B]}$ that constructs the $\galg$-module with basis $\mathtt{B}$. It is instantiated with $\galg$ being a given integro-differential algebra and $\mathtt{B}$ the term monoid just described. We will equip this domain with the operations defined as below, using a functor named $\mathtt{IntDiffPol[\galg,K]}$. Later in this section we will present some sample computations.

\subsection{Canonical Forms for Integro-Differential Polynomials.}

It is clear that $K$-linear combinations of terms of the form~\eqref{eq:sumnf} are still not canonical forms for the integro-differential polynomials since by the linearity of the integral, terms like
\begin{displaymath}
f \cum (g+h) u \quad \text{and} \quad f \cum g u + f \cum h u
\end{displaymath}
\noindent or terms like
\begin{displaymath}
f\cum \lambda g u \quad \text{and} \quad \lambda f \cum g u
\end{displaymath}
with $f, g, h \in \galg$ and $\lambda \in K$ represent the same polynomial. To solve this problem, we can consider terms of the form~\eqref{eq:sumnf} modulo these identities coming from \emph{linearity} in the ``coefficient'' $f$ and the integral, in analogy to the ideal $\basis$ introduced in Sect.~\ref{sec:intdiffop} for $\intdiffopbas$. Confer also~\cite{GuoKeigher2008}, where the tensor product is employed for constructing free objects in differential Rota-Baxter algebras. In the following, we assume for simplicity that $\galg$ is an ordinary integro-differential algebra.

More precisely, let $\canf$ denote the set of terms of the form~\eqref{eq:sumnf} and consider the free $K$-vector space generated by $\canf$. We identify terms
\begin{equation*}
f u(0)^\alpha u^\beta \cum f_1 u^{\gamma_1} \cum \cdots \cum f_n u^{\gamma_n}
\end{equation*}
with the corresponding basis elements in this vector space. Then we factor out the subspace generated by the following identities (analogous to the construction of the tensor product):
\begin{flalign*}
& f U \cum f_{1} u^{\gamma_1} \cum \cdots \cum (f_k+\tilde{f}_k) u^{\gamma_k} \cum \cdots\cum f_{n} u^{\gamma_n}\\
& \quad = f U \cum f_{1} u^{\gamma_1} \cum \cdots \cum f_k u^{\gamma_k} \cum \cdots\cum f_{n} u^{\gamma_n} + f U \cum f_{1} u^{\gamma_1} \cum \cdots \cum \tilde{f}_k u^{\gamma_k} \cum \cdots\cum f_{n} u^{\gamma_n}\\[1ex]
& f U \cum f_{1} u^{\gamma_1} \cum \cdots \cum (\lambda f_k) u^{\gamma_k} \cum \cdots\cum f_{n} u^{\gamma_n} = \lambda f U\cum f_{1} u^{\gamma_1} \cum \cdots \cum f_k u^{\gamma_k} \cum \cdots\cum f_{n} u^{\gamma_n}
\end{flalign*}
Here $U$ is short for $u(0)^{\alpha}u^{\beta}$, and there are actually two more identities of the same type for ensuring $K$-linearity in $f$. We write $\factor{\mathcal{R}}$ for this quotient space and denote the corresponding equivalence classes by
\begin{equation}
\label{eq:factorbasis}
\factor{f u(0)^\alpha u^\beta \cum f_1 u^{\gamma_1} \cum \cdots \cum f_n u^{\gamma_n}}.
\end{equation}
By construction, the quotient module $\factor{\canf}$ now respects the \emph{linearity relations}
\begin{flalign*}
& \factor{f U \cum f_{1} u^{\gamma_1} \cum \cdots \cum (f_k+\tilde{f}_k) u^{\gamma_k} \cum \cdots\cum f_{n} u^{\gamma_n}}\\
& \quad = \factor{f U \cum f_{1} u^{\gamma_1} \cum \cdots \cum f_k \cum \cdots\cum f_{n} u^{\gamma_n}} + \factor{f U \cum f_{1} u^{\gamma_1} \cum \cdots \cum \tilde{f}_k \cum \cdots\cum f_{n} u^{\gamma_n}}\\[1ex]
& \factor{f U \cum f_{1} u^{\gamma_1} \cum \cdots \cum (\lambda f_k) u^{\gamma_k} \cum \cdots\cum f_{n} u^{\gamma_n}} = \lambda \factor{f U \cum f_{1} u^{\gamma_1} \cum \cdots \cum f_k u^{\gamma_k} \cum \cdots\cum f_{n} u^{\gamma_n}}.
\end{flalign*}
together with the ones for linearity in $f$.

As for the tensor product, we have canonical forms for the factor space by expanding the ``coefficient'' $f$ and all the $f_k$ in~\eqref{eq:factorbasis} with respect to a $K$-\emph{basis} $\mathcal{B}$ for $\galg$, assuming $\mathcal{B}$ contains $1$. Then every polynomial can be written as a $K$-linear combination of terms of the form
\begin{equation}
    \label{eq:sumnfbasis}
    b u(0)^{\alpha}u^{\beta}\cum b_1 u^{\gamma_1} \cum \cdots \cum b_n u^{\gamma_n},
\end{equation}
where $b, b_1,\ldots,b_n \in \mathcal{B}$ with the condition on multi-indices as in Proposition~\ref{prop:surj}.

To show that terms of the form~\eqref{eq:sumnfbasis} are canonical forms for the integro-differential polynomials, we endow the quotient space $\factor{\canf}$ with an \emph{integro-differential structure} and invoke Proposition~\ref{prop:canforms}. For this we define the operations on the generators~\eqref{eq:factorbasis} and check that they respect the above linearity relations on $\factor{\canf}$.

First, we define a \emph{multiplication} on $\factor{\canf}$. Let $\fcanf \subseteq \canf$ denote the $K$-subspace generated by integral terms $\cum f_1 u^{\gamma_1} \cum \cdots \cum f_n u^{\gamma_n}$, including $1 \in \canf$ as the case $n=0$. Clearly, the nesting operation~\eqref{eq:nest} can be defined in a completely analogous manner on such integral terms (the only difference being that we have now derivatives of the indeterminate). Since it is clearly $K$-linear, it induces an operation $\nest\colon \factor{\fcanf} \times \factor{\fcanf} \rightarrow \factor{\fcanf}$. The next step is to define the shuffle product on $\fcanf$ just as in~\eqref{eq:shuffle2}, again with obvious modifications. Passing to the quotient yields the shuffle product $\,\cdot\,\colon \factor{\fcanf} \times \factor{\fcanf} \rightarrow \factor{\fcanf}$. This product is finally extending to a multiplication on all of $\factor{\canf}$ by setting
\begin{displaymath}
 \factor{f u(0)^{\alpha} u^{\beta} J} \factor{\tilde{f} u(0)^{\tilde{\alpha}} u^{\tilde{\beta}} \tilde{J}}= \factor{f \! \tilde{f} u(0)^{\alpha + \tilde{\alpha}} u^{\beta+\tilde{\beta}} (J \cdot \tilde{J})}
\end{displaymath}
where $J$ and $\tilde J$ range over $\fcanf$. Let us compute an example:

\begin{mmacode}
\includegraphics[scale=0.75]{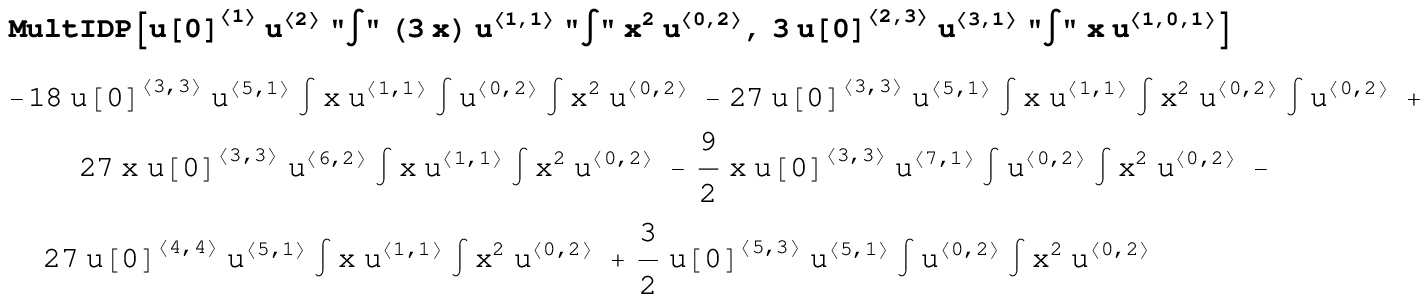}
\end{mmacode}

\noindent Since the multiplication on $\galg$ and the shuffle product are commutative, associative, and distributive over addition, the multiplication on $\factor{\canf}$ is well-defined and gives $\factor{\canf}$ the structure of a commutative $K$-algebra.

The definition of a \emph{derivation} $\der$ on this algebra is straightforward, using the fact that it should respect $K$-linearity and the Leibniz rule (treating also the $u(0)^\alpha$ as constants), that it should restrict to the derivation on differential polynomials (which in turn restricts to the derivation on $\galg$), and finally that it should also satisfy the section axiom~\eqref{eq:section-axiom}. Here is a sample computation:

\begin{mmacode}
\includegraphics[scale=0.75]{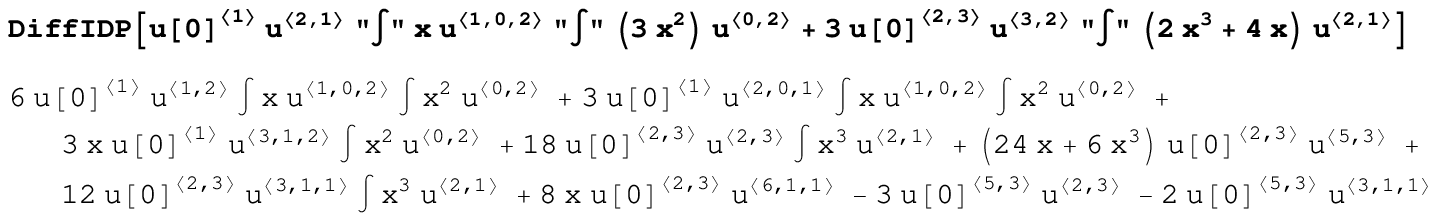}
\end{mmacode}

\noindent Using the $K$-linearity of this derivation, one verifies immediately that it is well-defined. From the definition it is clear that $K$-linear combinations of generators of the form $[u(0)^\alpha]$ are constants for $\der$, and one can also check that all constants are actually of this form.

Finally, we define a $K$-linear \emph{integral} on the differential $K$-algebra $(\factor{\canf},\der)$. Since we have to distinguish three different types of integrals, here and subsequently we will use the following notation: the usual big integral sign $\displaystyle\integral$ for the integration to be defined, the small integral sign $\cum$ for the elements of $\canf$ as we have used it before, and $\intf$ for the integral on $\galg$.

The definition of the integral on $\factor{\canf}$ is recursive, first by depth and then by order of $u^\beta$, following the classification of monomials from Definition~\ref{def:class}. In the \emph{base case} of zero depth and order, we put
\begin{equation}
\label{eq:int_base}
\integral \factor{f u(0)^\alpha} = \factor{\intf f} \factor{ u(0)^\alpha}.
\end{equation}
Turning to \emph{quasiconstant} monomials, we use the following definition (which actually includes the base case when $J=1$):
\begin{equation}
\label{eq:int_quasiconst}
  \integral \factor{f u(0)^\alpha J} =
  \factor{u(0)^\alpha (\intf f) J} -
  \factor{u(0)^\alpha \cum (\intf f) J'}.
\end{equation}

\noindent In the \emph{quasilinear} case we write the generators in form
\begin{displaymath}
 \factor{f u(0)^{\alpha} V u_k^{\beta_k} u_{k+1} J} \quad \text{with}
 \quad V= u_0^{\beta_0} \cdots u_{k-1}^{\beta_{k-1}}
\end{displaymath}
and construct the integral via~(\ref{eq:diff_baxter_terms}). Writing $s = \beta_k +1$, we have $u_k^{\beta_k} u_{k+1} = (u_k^s)'/s$, so we can define
\begin{equation}
\label{eq:int_quasilin}
  \integral \factor{f u(0)^{\alpha} V (u_k^s)' J}
  = \factor{f u(0)^{\alpha} V u_{k}^s J} - \factor{u(0)^{\alpha}} \integral\factor{f V J}' \factor{u_k^s} -  f(0) \, \factor{u(0)^{\alpha+\beta} \kron{J} },
\end{equation}
where we write $f(0)$ for $\evl(f)$ and $\kron{J}$ is $1$ for $J=1$ and zero otherwise. In the \emph{functional} case, we set
\begin{equation}
\label{eq:int_func}
   \integral \factor{f u(0)^\alpha u^\beta J} = \factor{u(0)^\alpha \cum
     f u^\beta J},
\end{equation}
so here we can just let the integral sign slip into the equivalence class. One may check that the integral is well-defined in all the cases by an easy induction proof, using $K$-linearity of the integral, the evaluation on $\galg$, and the derivation on $\factor{\canf}$.

Here is a small \emph{example} of an integral computed in the quasiconstant case (note that $\mathtt{IntIDP}$ corresponds to the big integral and \texttt{"}$\mathtt{\cum}$\texttt{"} to $\cum$ in our notation):
\begin{mmacode}
\includegraphics[scale=0.75]{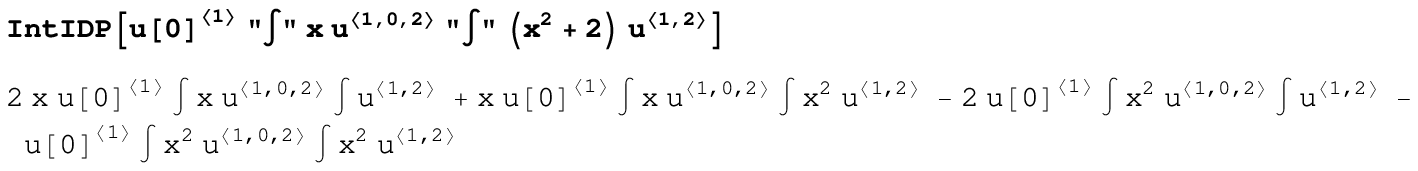}
\end{mmacode}
The next example computes an integral in the quasilinear case:
\begin{mmacode}
\includegraphics[scale=0.75]{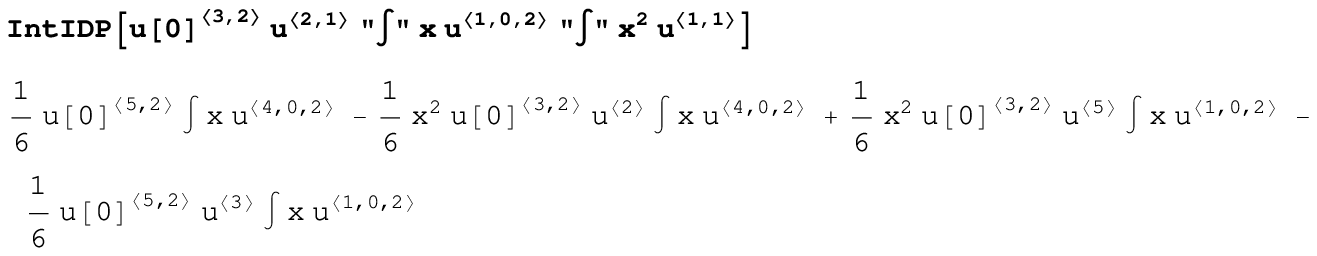}
\end{mmacode}
Note that all differential monomials within integrals are functional again, as it must be by our definition of~$\factor{\canf}$.

By construction the integral defined above is a section of the derivation on $\factor{\canf}$. So for showing that $\factor{\canf}$ is an \emph{integro-differential algebra} with operations, it remains only to prove the differential Baxter axiom~\eqref{eq:diff-baxter-axiom}. Equivalently, we can show that the evaluation
\[
\evl=1-\int \der
\]
is multiplicative by Corollary~\ref{cor:intdiffalg-ints}.

Recall that the algebra of constants $\const$ in $(\factor{\canf},\der)$ consists of $K$-linear combinations of generators of the form $\factor{u(0)^\alpha}$. By a short induction proof, we see that
\begin{equation}
 \label{eq:homogen_const}
 \integral \factor{u(0)^\alpha} \factor{R} = \factor{u(0)^\alpha} \integral
 \factor{R}.
\end{equation}
Hence the integral is homogeneous over the constants.

For showing that the \emph{evaluation} is multiplicative, we first reassure ourselves that it operates in the expected way on integro-differential monomials.
\begin{lemma}
\label{lem:eval_mon}
We have
\[
\evl \, \factor{f u(0)^{\alpha} u^{\beta} J} = f(0)\,\factor{u(0)^{\alpha+\beta} \kron{J}},
\]
where $\kron{J}$ is $1$ for $J=1$ and zero otherwise as in~\eqref{eq:int_quasilin}.
\end{lemma}
\begin{proof}
Note that $\evl$ is $\const$-linear by~\eqref{eq:homogen_const}, so we can omit the factor $u(0)^{\alpha}$. Assume first $\beta=0$. Then by the quasiconstant case~\eqref{eq:int_quasiconst} of the definition of the integral, we have
\[
\evl \, \factor{f J} = \factor{fJ} - \int \factor{fJ}' = \factor{fJ} - \factor{(\intf f')J} + \int \factor{(\intf f') J'} - \int \factor{fJ'},
\]
which by $\intf f' = f - f(0)$ gives
\[
f(0) \, \factor{J} - f(0) \int \factor{J}' = f(0) \factor{\kron{J}}
\]
because
\[
 \int \factor{J}' = \factor{J}\quad\text{for} \quad J\neq 1
\]
by the functional case~\eqref{eq:int_func} and zero for $J=1$. If $\beta \neq 0$ is of order $k$, we write $u^\beta=V u_k^s$ with $s\neq 0$, and we compute
\[
\evl \, \factor{f u^\beta J} = \factor{f V u_k^s  J} - \int \factor{f V u_k^s J}'=
f(0)\, \factor{u(0)^\beta \kron{J}}
\]
by the quasilinear case~\eqref{eq:int_quasilin} and the Leibniz rule.
\qed
\end{proof}

\begin{theorem}
\label{thm:baxter}
With the operations defined as above, $(\factor{\canf},\der,\smash{\displaystyle \int})$ has the structure of an integro-differential algebra.
\end{theorem}
\begin{proof}
As mentioned above, it suffices to prove that $\evl$ is multiplicative, and we need only do this on the generators. Again omitting the $u(0)^{\alpha}$, we have to check that
\[
\evl \, \factor{f  u^{\beta} J} \factor{\tilde f  u^{\tilde \beta} \tilde J} = \evl \factor{f \tilde f \,  u^{\beta + \tilde \beta } \, ( J \cdot \tilde J) } =\evl\factor{f  u^{\beta} J} \cdot \evl \factor{\tilde f  u^{\tilde \beta} \tilde J}.
\]
The case $J=\tilde J =1$ follows directly from Lemma~\ref{lem:eval_mon} and the multiplicativity of $\evl$ in $\galg$. Otherwise the shuffle product $J \cdot \tilde J$ is a sum of integral terms, each of them unequal one. Using again Lemma~\ref{lem:eval_mon} and the linearity of $\evl$, the evaluation of this sum vanishes, as does
$\evl\factor{f  u^{\beta} J} \cdot \evl \factor{\tilde f  u^{\tilde \beta} \tilde J}$.
\qed
\end{proof}

Since $\factor{\canf}$ is an integro-differential algebra, we can conclude by Proposition~\ref{prop:surj} and Proposition~\ref{prop:canforms} that $\factor{\canf}$ leads to canonical forms for integro-differential polynomials, up to the linearity relations: After a choice of basis, terms of the form~\eqref{eq:sumnfbasis} constitute a system of \emph{canonical forms} for $\intdiffpol{u}$. In the \tma\ implementation, we actually compute in $\factor{\canf}$ and do basis expansions only for deciding equality.

\section{From Rewriting to Parametrized Gr\"obner Bases}
\label{sec:rwtogb}

Equipped with the integro-differential polynomials, we can now tackle the task of proving the convergence of the reduction rules in Table~\ref{fig:red-rules}. As explained in Sect.~\ref{sec:intdiffop}, we will use the \emph{Diamond Lemma} (Theorem~\ref{thm:diamond-lemma}) for this purpose. First of all we must therefore construct a noetherian monoid ordering $>$ on $\free$ that is compatible with the reduction rules. In fact, there is a lot of freedom in defining such a $>$. It is sufficient to put $\der > f$ for all $f \in \galg$ and extend this to words by the graded lexicographic construction. The resulting partial ordering is clearly noetherian (since it is on the generators) and compatible with the monoid structure (by its grading). It is also compatible with the rewrite system because all rules reduce the word length except for the Leibniz rule, which is compatible because $\der > f$.

Thus it remains to prove that all ambiguities of Table~\ref{fig:red-rules} are resolvable, and we have to compute the corresponding S-polynomials and reduce them to zero. On the face of it, there are of course \emph{infinitely many} of these, suitably parametrized by $f,g \in \galg$ and $\phi, \psi \in \Phi$. For example, let us look at the minimal fork generated by $\cum u \cum v \cum$. In this case, the rule $\cum f \cum$ may be applied either with $f = u$ or with $f = v$ yielding the reductions
\begin{equation*}
\begin{array}[c]{rcl}
&\smash{\cum u \cum v \cum}\\
\swarrow && \searrow\\[1ex]
(\cum \act u) \, \cum v \cum - \cum \, (\cum \act u) v \cum &&
  \cum u \, (\cum \act v) \, \cum - \cum u \cum \, (\cum \act v)
\end{array}
\end{equation*}
with the S-polynomial $p = (\cum \act u) \, \cum v \cum - \cum \, (\cum \act u) v \cum - \cum u \, (\cum \act v) \, \cum + \cum u \cum \, (\cum \act v)$. But actually we should not call $p$ \emph{an} S-polynomial since it represents infinitely many: one for each choice of $u, v \in \galg$.

How should one handle this infinitude of S-polynomials? The problem is that for reducing S-polynomials like $p$ one needs not only the relations embodied in the reduction of Table~\ref{fig:red-rules} but also properties of the operations $\der, \cum\colon \galg \rightarrow \galg$ acting on $u,v\in\galg$. Since these computations can soon become unwieldy, one should prefer a method that can be automated. There are two options that may be pursued:
\begin{itemize}
\item Either one retreats to the viewpoint of \emph{rewriting}, thinking of Table~\ref{fig:red-rules} as a two-level rewrite system. On the upper level, it applies the nine parametrized rules with $f,g\in\galg$ and $\phi,\psi\in\Phi$ being arbitrary expressions. After each such step, however, there are additional reductions on the lower level for applying the properties of $\der, \cum\colon \galg \rightarrow \galg$ on these expressions. Using a custom-tailored reduction system for the lower level, this approach was used in the old implementation for generated an automated confluence proof~\cite{Rosenkranz2005}.
\item Or one views an S-polynomial like $p$ nevertheless as a \emph{single} element, not in $\free$ but in $\freepol$ with $\newgalg = \intdiffpoltwo$. With this approach, one remains within the paradigm of \emph{parametrized \grb\ bases}, and the interlocked operation of the two levels of reduction is clarified in a very coherent way: The upper level is driven by the canonical simplifier on $\newintdiffop$, the lower level by that on $\intdiffpoltwo$.
\end{itemize}
It is the second approach that we will explain in what follows.

Using $\freepol$ instead of $\free$ takes care of the parameters $f,g \in \galg$ but then there are also the \emph{characters} $\phi,\psi\in\Phi$. The proper solution to this problem would be to use a refined version of integro-differential polynomials that starts from a whole family $(\cum_\phi)_{\phi\in\Phi}$ of integrals instead of the single integral $\cum$, thus leading to a corresponding family of evaluations $u(\phi)$ instead of the single evaluation $u(0)$. We plan to pursue this approach in a forthcoming paper. For our present purposes, we can take either of the following positions:
\begin{itemize}
\item The characters $\phi,\psi$ may range over an \emph{infinite} set $\Phi$, but they are harmless since unlike the $f,g\in\galg$ they do not come with any operations (whose properties must be accounted for by an additional level of reduction). In this case, Table~\ref{fig:red-rules} is still an infinitely generated ideal in $\freepol$, and we have to reduce infinitely many S-polynomials. But the ambiguities involving characters are all of a very simple nature, and their reduction of their S-polynomials is straightforward.
\item Alternatively, we may restrict ourselves to a \emph{finite} set of characters (as in most applications!) so that Table~\ref{fig:red-rules} actually describes a finitely generated ideal in $\freepol$, and we need only consider finitely many S-polynomials.
\end{itemize}
The second alternative is somewhat inelegant due to the proliferation of instances for rules like $\phi\psi \rightarrow \psi$. In our implementation, we have thus followed the first alternative with a straightforward treatment of parametrization in $\phi,\psi$ but we will ignore this issue in what follows.

We can now use the new \tma\ implementation for checking that the nine rules in Table~\ref{fig:red-rules} form a \grb\ basis in $\freepol$. As explained before, we use the Diamond Lemma for this purpose (note that the noetherian monoid ordering $>$ applies also to $\freepol$ except that we have now just two generators $u,v \in \newgalg = \intdiffpoltwo$ instead of all $f \in \galg$). Hence it remains to check that all S-polynomials reduce to zero. We realize this by using the appropriate functor hierarchy, as follows. We first build up the algebra of the integro-differential polynomials having, in turn, integro-differential polynomials as coefficients, via the functor
\[
\mathtt{IntDiffPolys[IntDiffPolys[\galg,K],K]}
\]
and we denote the resulting domain by $\mathbb{P}$. Then we consider an instance of the functor constructing the integro-differential operators over $\mathbb{P}$. Finally, the computations are carried out over the algebra created by the $\mathtt{GroebnerExtension}$ functor taking the latter instance as input domain, that allows to perform polynomial reduction, S-polynomials and the \grb\ basis procedure.

Of course, the S-polynomials are generated automatically, but as a concrete example we check the minimal fork considered above:

\begin{mmacode}
\includegraphics[scale=0.75]{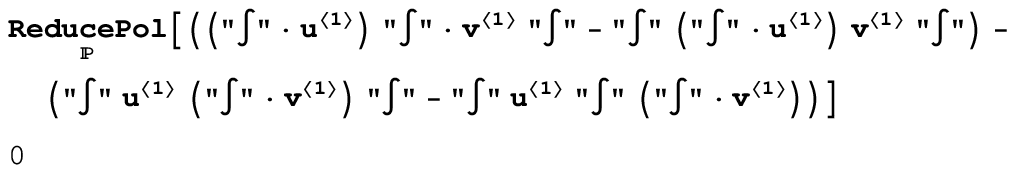}
\end{mmacode}

\noindent As it turns out, there are 17 nontrivial S-polynomials, and they all reduce to zero. This leads us finally to the desired convergence result for $\intdiffop$.

\begin{theorem}
The system of Table~\ref{fig:red-rules} represents a noncommutative \grb\ basis in $\free$ for any graded lexicographic ordering satisfying  $\der > f$ for all $f \in \galg$.
\end{theorem}
\begin{rrproof}
By the Diamond Lemma we must show that all S-polynomials $p \in \free$ reduce to zero. Since they may contain at most two parameters $f,g\in\galg$, let us write them as $p(f,g)$. But we have just observed that the corresponding S-polynomials $p(u,v) \in \freepol$ with $\newgalg = \intdiffpoltwo$ reduce to zero. Using the substitution homomorphism
\begin{equation*}
\phi\colon \newgalg \rightarrow \galg, \; (u,v) \mapsto (f,g),
\end{equation*}
lifted to $\newintdiffop \rightarrow \intdiffop$ in the obvious way, we see that $p(f,g) = \phi \, p(u,v)$ reduces to zero as well.
\end{rrproof}

\noindent From the conclusion of the Diamond Lemma, we can now infer that Table~\ref{fig:red-rules} indeed establishes a canonical simplifier for  $\intdiffop$.

\section{Conclusion}
\label{sec:conclusion}

The \emph{algebraic treatment of boundary problems} is a new development in Symbolic Analysis that takes its starting point in differential algebra and enriches its structures by introducing an explicit notion of integration and boundary evaluations. Recall the three basic tools that we have introduced for this purpose:
\begin{itemize}
\item The category of \emph{integro-differential algebras} $(\galg, \der, \cum)$ for providing a suitable notion of ``functions''. (As explained in Sect.~\ref{sec:theorema}, here we do not think of categories and functors in the sense of Eilenberg and Maclane---this is also possible and highly interesting but must be deferred to another paper.)
\item The functor creating the associated \emph{integro-differential operators} $\galg[\der, \cum]$ as a convenient language for expressing boundary problems (differential operators, boundary operators) and their solving Green's operators (integral operators).
\item The functor creating the associated \emph{integro-differential polynomials} $\intdiffpol{u}$, which describe the extension of an integro-differential algebra by a generic function $u$.
\end{itemize}
In each of these three cases, the differential algebra counterpart (i.e. without the ``integro-'') is well-known, and it appears as a rather simple substructure in the full structure. For example, the differential polynomials $\intdiffpol{u}$ over a differential algebra $(\galg, \der)$ are simple to construct since the Leibniz rule effectively flattens out compound terms. This is in stark contrast to an integro-differential algebra $(\galg, \der, \cum)$, where the Baxter rule forces the presence of nested integrals for ensuring closure under integration.

The interplay between these three basic tools is illustrated in a \emph{new confluence proof}: For an arbitrary integro-differential algebra $(\galg, \der, \cum)$, the rewrite system for the integro-differential operators $\intdiffop$ is shown to be a noncommutative \grb\ basis by the aid of the integro-differential polynomials $\intdiffpol{u,v}$. Having a confluent rewrite system leads to a canonical simplifier, which is crucial for the algorithmic treatment as expounded in Sect.~\ref{sec:theorema}.

Regarding our overall mission---the algebraic treatment of boundary problems and integral operators---we have only scratched the surface, and much is left to be done. We have given a brief overview of solving, multiplying and factoring boundary problems in Sect.~\ref{sec:bp}. But the real challenge lies ahead, namely how to \emph{extend our framework} to:
\begin{itemize}
\item \emph{Linear Boundary Problems for LPDEs}: As mentioned at the start of Sect.~\ref{sec:bp}, the algebraic framework for multiplying and factoring boundary problems is set up to allow for LPDEs; see~\cite{RegensburgerRosenkranz2009} for more details. But the problematic issue is how to design a suitable analog of $\intdiffop$ for describing integral and boundary operators (again the differential operators are obvious). This involves more than replacing $\der$ by $\der/\der_x, \der/\der_y$ and $\cum$ by $\cum_0^x$, $\cum_0^y$ because even the simplest Green's operators employ one additional feature: the transformation of variables, along with the corresponding interaction rules for differentiation (chain rule) and integration (substitution rule); see~\cite{RosenkranzRegensburgerTecBuchberger2009} for some first steps in this direction.
\item \emph{Nonlinear Boundary Problems:} A radically new approach is needed for that, so it seems appropriate to concentrate first on boundary problems for nonlinear ODEs and systems thereof. A natural starting point for such an investigation is the differential algebra setting, i.e. the theory of differential elimination~\cite{Hubert2003b,BoulierLazardOllivierEtAl1995,BoulierOllivierLazardPetitot2009}. By incorporating initial or boundary conditions, we can use explicit integral operators on equations, in addition to the usual differential operators (prolongations). As a consequence, the natural objects of study would no longer be differential but integro-differential polynomials.
\end{itemize}
We are well aware that such an approach will meet with many difficulties that will become manifest only as we progress. Nevertheless, we are confident that an algebraic---and indeed symbolic---treatment along these lines is possible.

\begin{acknowledgement}
We acknowledge gratefully the support received from the \emph{SFB F013} in Subproject~F1322 (principal investigators Bruno Buchberger and Heinz W.~Engl), in earlier stages also Subproject~F1302 (Buchberger) and Subproject~F1308 (Engl). This support from the Austrian Science Fund (FWF) was not only effective in its financial dimension (clearly a necessary but not a sufficient condition for success), but also in a ``moral'' dimension: The stimulating atmosphere created by the unique blend of symbolic and numerical communities in this SFB---in particular the Hilbert Seminar mentioned in Sect.~\ref{sec:intro}---has been a key factor in building up the raw material for our studies.

Over and above his general role in the genesis and evolution of the SFB F1322, we would like to thank \emph{Heinz W.~Engl} for encouragement, critical comments and helpful suggestions, not only but especially in the early stages of this project.

Loredana Tec is a recipient of a DOC-fFORTE-fellowship
of the Austrian Academy of Sciences at the Research Institute for Symbolic
Computation (RISC), Johannes Kepler University Linz. Georg Regensburger was partially supported by the Austrian Science Fund (FWF): J 3030-N18.

We would also like to thank an anonymous referee for giving us plenty of helpful suggestions and references that certainly increased the value of this article.

\end{acknowledgement}

%\bibliographystyle{spmpsci}
%\bibliography{references_cited}

\end{document}